\documentclass[twoside,11pt]{article}

\usepackage{enumerate}
\usepackage{amsfonts}
\usepackage{amsmath}
\usepackage{amssymb}
\usepackage{bm}
\usepackage{colonequals}
\usepackage{algorithm}
\usepackage{algorithmic}
\usepackage{tabularx}
\usepackage{booktabs}
\usepackage{multirow}
\usepackage{xspace}
\usepackage{xcolor}
\usepackage{caption}
\usepackage[abbrvbib,preprint]{jmlr2e}
\usepackage{hyperref}
\usepackage{url}
\usepackage{array}

\hypersetup{
  hidelinks
}

\makeatletter

\newcommand{\mrash}{{\rm Mr.ASH}\xspace}
\newcommand{\lasso}{{\rm Lasso}\xspace} 
\newcommand{\scad}{{\rm SCAD}\xspace} 
\newcommand{\mcp}{{\rm MCP}\xspace}
\newcommand{\lzero}{{\rm L0Learn}\xspace} 
\newcommand{\bayesb}{{\rm BayesB}\xspace}
\newcommand{\susie}{{\rm SuSiE}\xspace} 
\newcommand{\varbvs}{{\rm varbvs}\xspace}

\DeclareMathOperator*{\argmin}{argmin}
\DeclareMathOperator*{\argmax}{argmax}

\newcommand{\norm}[1]{\lVert #1\rVert}

\newcommand{\ba}{\[
	\begin{aligned}}
\newcommand{\ea}{
	\end{aligned}
	\]}
\newcommand{\baa}{\begin{equation}
	\begin{aligned}}
\newcommand{\eaa}{
	\end{aligned}
	\end{equation}}
\newcommand{\iid}{\overset{i.i.d.}{\sim}}

\newcommand{\tr}[1]{{\rm tr} (#1)}

\DeclareFontFamily{U}{mathx}{\hyphenchar\font45}
\DeclareFontShape{U}{mathx}{m}{n}{<-> mathx10}{}
\DeclareSymbolFont{mathx}{U}{mathx}{m}{n}
\DeclareMathAccent{\widebar}{0}{mathx}{"73}

\newcolumntype{M}[1]{>{\raggedright\arraybackslash}m{#1}}

\def\q{q}
\def\FNM{F^{\rm NM}}

\def\DKL{D_{\,{\rm KL}}}

\def\G{\mathcal{G}}
\def\Gfm{\G(\sigma_1^2,\dots,\sigma_K^2)}
\def\pk{\pi_k}
\def\Q{\mathcal{Q}}
\def\N{N}
\def\Gsmn{\G_{\mathrm{SMN}}}
\def\bhatbold{\hat{\b}}
\def\qhat{\hat{q}}
\def\ghat{\hat{g}}
\def\shat{\hat{\sigma}^2}

\def\rbar{\bar{\mathbf{r}}}
\def\pii{\bm{\pi}}
\def\b{\mathbf{b}}
\def\X{\mathbf{X}}
\def\y{\mathbf{y}}
\def\bj{b_j}
\def\bbar{\bar{\mathbf{b}}}

\def\shrinkf{S_{f,\sigma}}
\def\shrinkg{S_{g_\sigma,\sigma}}
\def\xj{\mathbf{x}_j}
\def\pprior{p_{\mathrm{prior}}}

\usepackage{lastpage}
\jmlrheading{25}{2024}{1-\pageref{LastPage}}{8/22; Revised
12/23}{6/24}{22-0953}{all authors FULL names}
\ShortHeadings{title}{all authors LAST names}

\ShortHeadings{Flexible empirical Bayes regression}{Kim, 
  Wang, Carbonetto and Stephens}
\firstpageno{1}

\begin{document}
	
\title{A Flexible Empirical Bayes Approach to Multiple Linear
  Regression, and Connections with Penalized Regression}
		
\author{\name Youngseok Kim \email youngseok@uchicago.edu \\
	\addr{Department of Statistics\\
	University of Chicago\\
	Chicago, IL 60637, USA}
	\AND
	\name Wei Wang \email weiwang@galton.uchicago.edu \\
	\addr{Department of Statistics\\
	University of Chicago\\
	Chicago, IL 60637, USA}
	\AND
	\name Peter Carbonetto \email pcarbo@uchicago.edu \\
	\addr{Research Computing Center and Department of Human Genetics\\
	University of Chicago\\
	Chicago, IL 60637, USA}
	\AND
	\name Matthew Stephens \email mstephens@uchicago.edu \\
	\addr{Department of Statistics and Department of Human Genetics\\
	University of Chicago\\
	Chicago, IL 60637, USA}}
\editor{Boaz Nadler}
\date{}

\maketitle
	
\begin{abstract}%
We introduce a new empirical Bayes approach for large-scale
multiple linear regression. Our approach combines two key ideas: (i)
the use of flexible ``adaptive shrinkage'' priors, which approximate
the nonparametric family of scale mixture of normal distributions by a
finite mixture of normal distributions; and (ii) the use of
variational approximations to efficiently estimate prior
hyperparameters and compute approximate posteriors. Combining
  these two ideas results in fast and flexible methods, with
  computational speed comparable to fast penalized regression methods
  such as the Lasso, and with competitive prediction accuracy across a
  wide range of scenarios.  Further, we provide new results that
  establish 
%
% strong 
%
conceptual connections between our empirical Bayes methods and
penalized methods. Specifically, we show that the posterior mean
from our method solves a penalized regression problem, with the form
of the penalty function being learned from the data by directly
solving an optimization problem (rather than being tuned by
cross-validation). Our methods are implemented in an R package, {\tt
  mr.ash.alpha}, available from
\url{https://github.com/stephenslab/mr.ash.alpha}.
\end{abstract}


\section{Introduction}

Multiple linear regression is one of the oldest statistical methods
for relating an outcome variable to predictor variables, dating back
at least to the eighteenth century \citep{stigler1984studies}. In
recent decades, data sets have grown rapidly in size, with the number
of predictor variables often exceeding the number of
observations. Fitting even simple models such as multiple linear
regression to large data sets raises interesting research
questions. These include questions about regularization and
  prediction (e.g., how to estimate the parameters to optimize
  out-of-sample prediction accuracy), and questions about variable
  selection and inference (e.g., how to choose the coefficients in the
  regression that are non-zero). In this paper, we focus on the
  former.
% although we also briefly mention variable selection and inference.

Many different approaches for regularization and prediction have been
proposed. Most fall into one of two types: penalized linear regression
(PLR) methods based on different penalized least-squares criteria
\citep[e.g.,][]{hoerl1970ridge, tibshirani1996regression,
  fan2001variable, miller2002subset, zou2005regularization,
  zhang2010nearly, hazimeh2018fast,trimmedlasso} and Bayesian
approaches based on different priors and computational methods
\citep[e.g.,][]{mitchell1988bayesian, george1993variable,
  meuwissen2001prediction, park2008bayesian, hans2009,
  carvalho2010horseshoe, li2010bayesian, griffin2010inference,
  guan-2011, habier2011extension, carbonetto2012scalable,
  zhou2013polygenic, drugowitsch2013variational, wang2018simple,
  rovckova2018spike, ray2019variational, zabad2023fast}.
  
These different approaches have different strengths and
weaknesses. For example, ridge regression \citep[an $L_2$-penalty on
  the coefficients;][]{hoerl1970ridge, tikhonov-1963} is simple,
involving a convex optimization problem and a single tuning parameter,
and performs well in ``dense'' settings (many predictors with non-zero
effects). However, it does not do well in ``sparse'' settings where
a small number of non-zero coefficients dominate. The Lasso \citep[an
  $L_1$-penalty on the coefficients;][]{tibshirani1996regression} is
similarly computationally convenient, and behaves better than ridge
regression in sparse settings. However, prediction accuracy of the
Lasso is limited by its tendency to ``overshrink'' large effects
\cite[e.g.,][]{su2017false}. The Elastic Net
\citep{zou2005regularization} combines some of the advantages of ridge
regression and the Lasso, and in its most general form includes both
as special cases; however, the Elastic Net also introduces an
additional tuning parameter that results in a non-trivial additional
computation expense.

Nonconvex penalties---examples include the $L_0$-penalty
\citep{miller2002subset, hazimeh2018fast}, the smoothly clipped
absolute deviation (SCAD) penalty \citep{fan2001variable} and the
minimax concave penalty (MCP) \citep{zhang2010nearly}---can also give
better prediction performance in sparse settings, but this comes with
the challenge of solving a nonconvex optimization problem.

Bayesian methods, by using flexible priors, have the potential to
achieve excellent prediction accuracy in both sparse and dense
settings
\citep[e.g.,][]{park2008bayesian,hans2009,griffin2010inference,
  li2010bayesian,guan-2011,zhou2013polygenic,zeng-2018}, but have some
practical drawbacks; notably, model fitting typically involves
performing Markov chain Monte Carlo (MCMC) with a potentially high
computational burden. Further, convergence of the Markov chain can be
difficult to diagnose, particularly for non-expert users. In summary,
when choosing among existing methods, one must confront tradeoffs
between prediction accuracy, flexibility and computational
convenience.

In this paper, we develop an approach to multiple linear regression
that aims to combine the best features of existing methods: it is
fast, comparable in speed to the cross-validated Lasso; it is
flexible, capable of adapting to sparse and dense settings; it is
self-tuning, with no need for user-specified hyperparameters; and, in
our numerical studies, its prediction accuracy was competitive with
the best methods against which we compared, in a wide range of
regression settings, including both dense and sparse
effects.

This consistent competitive performance across a wide range of
dense/sparse settings is particularly valuable in practice, because in
practical applications one does not know whether signals are dense or
sparse (or perhaps somewhere in between). Thus, in practice, users
will be reluctant to trust a method that does not perform consistently
well, even if it performs excellently (even optimally) in certain
settings. We believe that the consistently strong performance of our
method across settings, combined with its computational speed, make it
attractive for practitioners looking to apply large-scale multiple
linear regression to real problems.

% Further, we show that our method has a dual interpretation as both a
% penalized regression method and a Bayesian regression method,
% thereby providing a conceptual bridge between these two approaches.

Our method takes an {\em empirical Bayes} (EB) approach
\citep{robbins1964, efron2019oracle, hartleyrao1967,
carlin2000empirical, stephens2016false, casella2001empirical} to
multiple regression; that is, it assigns a prior to the coefficients
in the regression method, and this prior is learned from the
data. The EB approach is, in many ways, a natural approach for
attempting to attain the benefits of Bayesian methods while
addressing some of their computational challenges. Indeed, EB for
the multiple regression problem is far from new; for example, the
Bayesian Lasso \citep{park2008bayesian, blasso} takes an EB approach
with a Laplace prior in which the $\lambda$ parameter in the Laplace
prior is estimated by maximizing the (marginal) likelihood. Other EB
approaches use a normal prior \citep{nebebe1986bayes}, a
point-normal (``spike-and-slab'') prior
\citep{george2000calibration}, and a point-double-exponential prior
\citep{yuan2005efficient}. (See also \citealt{vandewiel2019learning}
for a review of other EB approaches.) However, previous EB
approaches to multiple regression have either focussed on relatively
inflexible priors, or have been met with considerable computational
challenges.

% For example, \cite{nebebe1986bayes} developed an EB approach with a
% normal prior---effectively ``EB ridge regression''--which makes
% computations easy, but is not well adapted to sparse settings.
% In contrast, \cite{george2000calibration} consider the point-normal
% prior (sometimes called a ``spike-and-slab'' prior), which is
% considerably more flexible, but makes computation difficult.
% \cite{george2000calibration} make several approximations, including
% use of ``conditional maximum likelihood'' (CML), which conditions on a
% single best model (that is, the subset of predictors with non-zero
% coefficients) instead of summing over all models as a conventional
% likelihood would. However, even the CML approximation is intractable,
% because finding the best model is intractable, and further
% approximations are required. Building on this work,
% \cite{yuan2005efficient} also uses the CML approximation to perform EB
% estimation for spike-and-slab priors, but they replace the normal slab
% with a Laplace (double-exponential) slab.
% To address computational difficulties, they reduce the model's
% flexibility; specifically, they constrain the two parameters of the
% prior to connect it to the Lasso in an interesting way, then they
% exploit this connection to make inferences more tractable.

Here, we propose a different EB approach that is both more flexible
than these previous EB approaches and more computationally
scalable. This new EB approach has two key components. First, to
increase flexibility, we borrow the ``adaptive shrinkage'' priors used
in \cite{stephens2016false}; specifically, we use the ``scale mixture
of normals'' priors. This prior family includes most of the popular
priors that have been used in Bayesian regression, including normal,
Laplace, point-normal, point-Laplace, point-$t$, normal-inverse-gamma,
Dirichlet-Laplace and horseshoe priors \citep{hoerl1970ridge,
  george1997approaches, meuwissen2001prediction, meuwissen2009fast,
  habier2011extension, dirichlet-laplace}. Increasing model
flexibility typically means greater computational expense, but in this
case the use of the adaptive shrinkage priors actually simplifies many
computations, essentially because the scale mixture family is a convex
family. Second, to make computations tractable, we adapt the
variational approximation methods for multiple regression from
\cite{carbonetto2012scalable}.
%
% Although these VA methods are approximations, they improve on the
% CML approximation because they sum over a large set of plausible
% models rather than simply selecting a single best model.
%
The main limitation of the variational approximation approach is that,
in sparse settings with very highly correlated predictors, it will
give only one of the correlated predictors a non-negligible
coefficient \citep{carbonetto2012scalable}. This limitation, which is
shared by several other existing methods, including the Lasso and
$L_0$-penalized regression, does not greatly affect prediction
accuracy. However, it does limit the conclusions that can be drawn
about the selected variables. Consequently, other methods
\cite[e.g.,][]{wang2018simple} may be preferred when the main goal is
variable selection for scientific interpretation rather than
prediction. Since our approach combines EB ideas with variational
approximations, we refer to it as a ``variational EB'' (VEB) approach.
 
While variational methods have previously been used to fit Bayesian
linear regression models \citep{girolami-2001, logsdon-2010,
  carbonetto2012scalable, wang2018simple, You2014, Ren2011}, they have
not been used to implement an EB method, and not with the flexible
class of priors we consider here. (Independently,
\citealt{zabad2023fast} recently used a VEB approach, but with less
flexible priors.) Our work arises from the combination of two earlier
ideas: (1) the use of variational approximation techniques for fast
posterior computation in large-scale linear regression
\citep{carbonetto2012scalable}; and (2) the use of a flexible class of
priors \eqref{def:finitemixture}, which was originally proposed in
\cite{stephens2016false} for performing EB inference in a simpler, but
related, problem: the ``normal means problem'' \citep{efron-1973,
  johnstone2004needles, sun2018solving, castillo2012needles,
  bhadra2019lasso}. The combination of these two ideas results in
methods that are simpler, faster, more flexible, and often more
accurate than those in \cite{carbonetto2012scalable}.

% While our work is closely connected to \cite{carbonetto2012scalable},
% there are two key differences compared with that previous work: our
% work replaces the point-normal prior with the more flexible normal
% scale mixture family, $\Gfm$; and it replaces the numerical
% integration over $g, \sigma^2$ with a maximization over $g$ and
% $\sigma^2$ (as is conventional in EB methods). These modifications
% lead to substantial algorithmic simplifications and computational
% speedups without sacrificing prediction performance (see
% Section~\ref{sec:experiment}).

Finally, another key contribution of our paper is to provide
a conceptual bridge between PLR methods and Bayesian
methods. Specifically, we show that our VEB approach is actually a
PLR method, where the penalty function is learned from the data by
{\em directly solving an optimization problem} rather than being
tuned by cross-validation (CV). This result is not only conceptually
interesting, but opens the door to other potential algorithms (e.g.,
gradient descent) for the VEB problem. Tuning multiple parameters by
solving an optimization problem is also more practical than CV; for
example, our VEB approach has a similar computational cost to
methods such as the Lasso that tune a single parameter by CV, and is
substantially faster than methods such as the Elastic Net that tune
two or more parameters by CV.

% In particular, {\em we show that the algorithm can be viewed as a
% coordinate-ascent algorithm for fitting a PLR.}  %\changed{[Note
% that, while some PLR methods such as the Spike-and-Slab Lasso
% \citep{rovckova2018spike, bai2021spike, scheipl2012spikeandslab}
% have been inspired by Bayesian approaches, none of these PLR
% methods, as far as we know, are formally understood to be Bayesian
% methods.]}  In contrast to existing PLR methods which assume a
% relatively restrictive class of penalty functions, with usually just
% one or two parameters that are tuned by cross-validation, our VEB
% approach has a much more flexible family of penalty functions
% (corresponding to our flexible family of priors), and the form of
% the penalty function is adapted to the data.  % (by solving an
% optimization problem, analogous to how EB learns % priors from data
% by maximizing the likelihood).  While one might expect that this
% flexibility comes at a greater computational cost,

% Our methods are implemented as an R package, {\tt
%   mr.ash.alpha} (``Multiple Regression with Adaptive SHrinkage
% priors''), available online at
% \url{https://github.com/stephenslab/mr.ash.alpha}.

\subsection{Organization of the Paper}

The remainder of the paper is organized as follows.
Section~\ref{sec:background} gives preliminary background on the
normal means model and introduces some notation.
Section~\ref{sec:mrashmodel} describes our VEB methods and
optimization algorithms in detail. Section~\ref{sec:penalty} makes
connections between our VEB approach and penalized
approaches. Section~\ref{sec:experiment} gives results from numerical
studies comparing prediction performance of different methods for
multiple linear regression, including our VEB approach.
Section~\ref{sec:conclusion} summarizes the contributions of this work
and discusses future directions.

\subsection{Notations and Conventions}

We write vectors in bold, lowercase letters (e.g., $\mathbf{b}$) and
matrices are written in bold, uppercase letters (e.g., $\X$) We use
$\mathbb{R}^n$ to denote the set of real-valued vectors of length $n$,
$\mathbb{R}_{+}^n$ for the set of real non-negative vectors of length
$n$, $\mathbb{R}^{m \times n}$ for the set of real $m \times n$
matrices, and $\mathbb{S}^n = \{ \mathbf{x} \in \mathbb{R}_{+}^{n} :
\sum_{i=1}^n x_i = 1 \}$ denotes the $n$-dimensional simplex. We use
$\mathbf{x}_j$ to denote the $j$th column of matrix $\X$. We write
sets and families in calligraphic font, e.g., $\mathcal{G}$. We use
$N({\bf x}; {\bm\mu}, {\bm\Sigma})$ to denote the probability density
of the multivariate normal distribution at ${\bf x} \in \mathbb{R}^n$
with mean ${\bm\mu} \in \mathbb{R}^n$ and $n \times n$ covariance
matrix ${\bm\Sigma}$. We use ${\bf I}_n$ to denote the $n \times n$
identity matrix. We use $\norm{\mathbf{x}} = \sqrt{\mathbf{x}^T
  \mathbf{x}}$ to denote the $L_2$-norm of vector $\mathbf{x} \in
\mathbb{R}^n$, and $\|\mathbf{x}\|_{\infty} =
\max\{|x_1|, \ldots, |x_n|\}$ denotes the $L$-infinity norm of
$\mathbf{x}$. We use $\triangleq$ to indicate definitions.

\section{Preliminaries: the Empirical Bayes Normal Means Model}
\label{sec:background} 
\label{subsec:ebnm}

The computations for our approach are closely related to those for a
simpler model known as the ``normal means'' (NM) model. This section
reviews this model and introduces notation that will be used later.

\subsection{The Normal Means Model}

The normal means model is a model for a sequence $y_1, \dots, y_p$ of
observations in which each observation $y_j$ is normally distributed
with unknown mean $\bj$ and known variance $\sigma^2$:
\begin{equation} 
\label{model:normal_means}
y_j \mid \bj, \sigma^2 \sim N(\bj, \sigma^2), 
\quad j=1, \dots, p.
\end{equation}
This can be viewed as a special case of multiple linear regression in
which the covariates are orthogonal and the residual variance is known.
(Specifically, it is equivalent to \eqref{model:bvs} below with $\X =
{\bf I}_n$ and $\sigma^2$ known.)

\subsection{The Normal Means Model with Adaptive Shrinkage Priors}

\cite{stephens2016false} considers an EB version of the NM model that
assumes the $b_j$ are {\em i.i.d.} from some prior distribution $g$
that is to be estimated from the observed data. Specifically,
\cite{stephens2016false} considers priors that are scale mixtures of
normals; that is, $g \in \Gfm$, where
\begin{equation}
\label{def:finitemixture}
\Gfm \triangleq 
\bigg\{g = \sum_{k=1}^K \pi_k N(0, \sigma_k^2) :
\mathbf{\pi} \in \mathbb{S}^K \bigg\},
\end{equation}
and where $0 \leq \sigma_1^2 < \cdots < \sigma_K^2 < \infty$ is a
pre-specified grid of component variances, and $\pi_1, \dots, \pi_K$
are unknown mixture proportions. Typically, $\sigma_1^2=0$ so that
$\Gfm$ includes sparse prior distributions. (We define $N(0, 0)$ to be
the Dirac ``delta'' mass at zero, commonly denoted as $\delta_0$.) By
making the grid of variances sufficiently wide and dense, the prior
family $\Gfm$ can approximate, with arbitrary accuracy, the
nonparametric family of all the scale mixtures of zero-mean normal
distributions. This nonparametric family, which we denote by $\Gsmn$,
is very flexible, and includes most popular distributions used as
priors in Bayesian regression models, including normal (``ridge
regression'') \citep{hoerl1970ridge}, point-normal (``spike and
slab'') \citep{chipman2001practical, george1993variable,
  george1997approaches, mitchell1988bayesian}, double-exponential or
Laplace \citep{Figueiredo2003, park2008bayesian, hans2009,
  tibshirani1996regression, li2010bayesian}, horseshoe
\citep{carvalho2010horseshoe}, normal-gamma prior
\citep{griffin2010inference}, normal-inverse-gamma prior
\citep{meuwissen2009fast, habier2011extension}, mixture of two normals
(BSLMM) \citep{zhou2013polygenic}, and the mixture of four
zero-centered normals with different variances suggested by
\cite{moser2015simultaneous}.

\cite{stephens2016false} refers to the priors
\eqref{def:finitemixture} as ``adaptive shrinkage'' priors. Here, we use
these adaptive shrinkage priors, but assume a prior distribution in
which the {\em scaled} coefficients, $\bj/\sigma$, are {\em i.i.d.}
from $g$,
\begin{equation}
\label{model:prior}
\bj \mid g, \sigma^2 \, \iid \, g_{\sigma},
\end{equation}
where $g_{\sigma}(\bj) \triangleq g(\bj/\sigma)/\sigma$. We use this
prior on the scaled coefficients because in the regression setting it
may help to reduce issues with multi-modality; see
\citealt{park2008bayesian} for example. The scaled prior
\eqref{model:prior} also provides computational benefits in the
``fully Bayesian'' regression setting; see, for example,
\citealt{chipman2001practical, george1997approaches, liang-2008} for
arguments in favour of the scaled prior. All our methods can also be
applied, with minor modifications, to work with the unscaled prior
$\bj \iid g$.

\subsubsection{Augmented-variable Representation}

It is helpful to think of $g \in \G(\sigma_1^2, \dots, \sigma_K^2)$ as
determining a set of mixture proportions ${\bm \pi} = (\pi_1, \dots,
\pi_K)$ and component variances $\sigma_1^2, \dots, \sigma_K^2$ of a
normal mixture. When $g \in \G(\sigma_1^2, \dots, \sigma_K^2)$, the
prior \eqref{model:prior} can also be written as
\begin{equation} 
\label{eqn:normalmixture2}
b_j \mid g,\sigma^2 \iid\, 
\sum_{k=1}^K \pi_k N(0, \sigma^2\sigma_k^2).
\end{equation}
It is also convenient for some of the derivations to introduce the standard
augmented-variable representation of this mixture:
\begin{equation}
\label{eqn:latent}
\begin{aligned}
p(\gamma_j = k \mid g) &= \pi_k \\
b_j \mid g, \sigma^2, \gamma_j = k &\sim N(0, \sigma^2\sigma_k^2),
\end{aligned}
\end{equation}
where the latent variable $\gamma_j \in \{1, \dots, K\}$ indicates
which mixture component gave rise to $b_j$.

\subsection{Empirical Bayes for the Normal Means Model}

\cite{stephens2016false} provides EB methods to fit the normal means
model. These methods proceed in two steps: estimate $g$ (Step 1);
compute the posterior distribution for $\b$ given the estimated $g$
(Step 2). Step 1 is simplified by the use of a fixed grid of variances
in \eqref{eqn:normalmixture2}, which means that only the mixture
proportions ${\bm\pi}$ need to be estimated. This is done by
maximizing the marginal log-likelihood:
% \begin{equation}
% p^{\rm NM}(\y \mid g, \sigma^2) = \prod_{j=1}^p \sum_{k=1}^K \pi_k L_{jk},
% \end{equation}
\begin{align}
\hat{\bm\pi} &= \argmax_{\bm\pi \,\in\, \mathbb{S}^K}
\, \log p(\y \mid g, \sigma^2) \nonumber \\
&= 
\argmax_{\bm\pi \,\in\, \mathbb{S}^K}
\sum_{j=1}^p \log \sum_{k=1}^K \pi_k L_{jk},
\label{eqn:pihat}
\end{align}
where 
\begin{align} 
L_{jk} &\triangleq
p(y_j \mid g, \sigma^2, \gamma_j = k) \nonumber \\
&= N(y_j; 0, \sigma^2 + \sigma^2\sigma_k^2),
\label{eqn:L}
\end{align}
and $\y \triangleq (y_1, \ldots, y_p)$. This is a convex optimization
problem, and can be solved efficiently using convex optimization
techniques \citep{koenker2014convex, kim2018fast}, or simply by
iterating the following Expectation Maximization (EM) updates
\citep{dempster1977maximum}:
\begin{align}
\mbox{E-step} &\quad \phi_{jk} \leftarrow
\phi_{k}(y_j; g, \sigma^2) \triangleq
p(\gamma_j = k \mid y_j, g, \sigma^2) =
\frac{\pi_k L_{jk}}{\sum_{k'=1}^K \pi_{k'} L_{jk'}},
\label{eqn:NM-E} \\
\mbox{M-step} &\quad \pi_k \leftarrow
\frac{1}{p} \sum_{j=1}^p \phi_{jk},\quad k = 1, \dots, K.
\label{eqn:NM-M}
\end{align}
The posterior mixture assignment probabilities $\phi_{jk}$ are
sometimes referred to as the ``responsibilities''.

Step 2, computing the posterior distribution, is also straightforward,
again due to the independence of the observations and the conjugacy of
the normal (mixture) prior with the normal likelihood:
\begin{align} 
p_{\rm post}^{\rm NM}(\bj, \gamma_j=k \mid y_j, g, \sigma^2) &=
p(b_j \mid y_j, g, \sigma^2, \gamma_j=k) \,
p(\gamma_j = k \mid y_j, g, \sigma^2) 
\nonumber \\
&= \phi_{jk} \, N(b_j; \mu_{jk}, s_{jk}^2),
\label{def:posterior_mixture}
\end{align}
where
\begin{align} 
\mu_{jk} \triangleq \mu_k(y_j;g,\sigma^2) &=
{\textstyle \frac{\sigma_k^2}{1+\sigma_k^2}} \times y_j, \label{def:muk} \\
s_{jk}^2 \triangleq s_k^2(y_j;g,\sigma^2) &=
{\textstyle \frac{\sigma_k^2}{1+\sigma_k^2}} \times \sigma^2. \label{def:sk}
\end{align}
(Although the posterior variances do not depend on $y_j$, we write
them as $s_k^2(y_j; g, \sigma^2)$ for notational consistency.)  Summing
the component posterior \eqref{def:posterior_mixture} over $k$ then
yields an analytic expression for the posterior of $\bj$,
\begin{equation} \label{eqn:post_ebnm}
p_{\rm post}^{\rm NM}(\bj \mid y_j, g, \sigma^2) =
\sum_{k=1}^K \phi_{jk} \, N(b_j; \mu_{jk}, s_{jk}^2).
\end{equation}

\section{Variational Empirical Bayes Linear Regression}
\label{sec:mrashmodel}

\subsection{Empirical Bayes Linear Regression}
\label{subsec:eblr}

We consider the multiple linear regression model,
\begin{equation}
\label{model:bvs}
\y \mid \X, \b, \sigma^2 \sim N(\X\b, \sigma^2 {\bf I}_n),
\end{equation}
where $\y \in \mathbb{R}^n$ is a vector of responses, ${\bf X} \in
\mathbb{R}^{n \times p}$ is a design matrix whose columns contain
predictors ${\bf x}_1, \ldots, \mathbf{x}_p \in \mathbb{R}^n$, $\b \in
\mathbb{R}^p$ is a vector of regression coefficients, and $\sigma^2
\geq 0$ is the variance of the residual errors. While an intercept is
not explicitly included in \eqref{model:bvs}, it is easily accounted
for by centering $\y$ and the columns of $\X$ prior to model fitting
\citep{chipman2001practical}; see also
Section~\ref{sec:practicalconsiderations}. To simplify presentation,
we will assume throughout the main text of the paper that the columns
of $\X$ are rescaled so that $\norm{\xj} = 1$, for $j = 1, \dots, p$.
However, all our methods and results can be extended to the unscaled
case; Appendix~\ref{appendix:derivation} includes these
extensions.

Taking an EB approach, as in the NM model above, we assume the scaled
regression coefficients, $b_j/\sigma$, are {\em i.i.d} from some prior
$g$, where $g$ is to be estimated from the observed data. Although our
methods apply more generally, we focus on the adaptive shrinkage
priors \eqref{eqn:normalmixture2} because they are flexible and
computationally convenient.

A standard EB approach to fitting the regression model \eqref{model:bvs}
with priors \eqref{model:prior} would, similar to above, involve
the following two steps:
\begin{enumerate}

\item Estimate $g, \sigma^2$ by maximizing the marginal likelihood:
\begin{align}
(\hat{g}, \hat{\sigma}^2) &=
\argmax_{g \,\in\, \mathcal{G}, \,\sigma^2 \,\in\, \mathbb{R}_{+}}
p(\y \mid \X, g, \sigma^2) \nonumber \\
&= \argmax_{g \,\in\, \mathcal{G},\, \sigma^2 \,\in\, \mathbb{R}_{+}}
\log {\textstyle \int p(\y \mid \X, \b, \sigma^2) \,
     p(\b \mid g, \sigma^2) \, d\b.}
\label{eqn:prior_estimation}
\end{align}

\item Infer $\b$ based on the posterior distribution, 
\begin{equation}
\label{eqn:posterior_calculation}
\hat{p}_{\rm post}(\b) \triangleq 
p(\b \mid \X, \y, \hat{g},\hat{\sigma}^2)
  \propto p(\y \mid \X, \b, \hat{\sigma}^2) \,
  p(\b \mid \hat{g}, \hat{\sigma}^2).	
\end{equation}
% In particular, for prediction one typically estimates $b_j$ by its
% posterior mean, 
% \ba
% \mathbb{E}(\bj \mid \X, \y, \ghat, \shat) =
% \mathbb{E}_{p_{\rm post}}(b_j).
% \ea
 \end{enumerate}
Unfortunately, in contrast to the NM model, both steps are
computationally impractical due to intractable integrals or very large
sums, or both, except in special cases.

% For parameter estimation (Step 1), the optimization is hard due
% to the intractable, high-dimensional integral in the marginal
% likelihood. For posterior inference (Step 2), computation of the
% posterior \eqref{eqn:posterior_calculation} involves an intractable
% normalization constant (also the marginal likelihood), and computation
% of the posterior expectations may involve additional intractable
% integrals.

\subsection{Variational Approximation}
\label{subsec:variational_background}

To circumvent the intractability of the EB approach, we use a
mean-field variational approximation \citep{blei2017variational,
  jordan1999introduction, wainwright2008graphical,logsdon-2010,
  carbonetto2012scalable} to derive a ``variational empirical Bayes''
(VEB) approach. The idea of VEB inference is mentioned explicitly in
\cite{lda}, although earlier work implemented similar ideas
(e.g., \citealt{saul-1996, ghahramani-2000}; see also
\citealt{vandewiel2019learning}).
%
% These VEB methods end up solving Steps 1 and 2 above {\em
%   iteratively} because the variational approximation introduces an
% interdependency between the parameter estimation (Step 1) and
% posterior computation (Step 2). This iterative procedure can be
% understood as solving a single optimization problem (described in
% Section~\ref{sec:mrashmodel}), and is a version of the generalized EM
% framework of \cite{neal-hinton-1998}. Because of their close
% connection to EM, such algorithms are sometimes referred to as
% ``variational EM'' algorithms, although this terminology can be
% confusing as it has been used in other ways; for example, the
% ``variational Bayesian EM'' of \citet{beal2003variational} is not the
% same as what we use here.
%
To describe the VEB approach, it is convenient to rewrite the two
steps of EB as solving a single optimization problem (see
also the Supplementary Materials from \citealt{wang2018simple}):
\begin{equation} 
\label{eq:EBsingle}
(\hat{p}_{\rm post}, \hat{g}, \hat\sigma^2) = 
\argmax_{q,\, g \,\in\, \mathcal{G}, \,\sigma^2 \,\in\, \mathbb{R}_{+}}
F(q, g, \sigma^2),
\end{equation}
where the optimization over $q$ is over all possible distributions on
$(\b, {\bm\gamma})$, and
\begin{equation}
F(q,g,\sigma^2) \triangleq \log p(\y \mid \X, g, \sigma^2) -
\DKL(q(\b, {\bm\gamma}) \,\|\, p(\b, {\bm\gamma} \mid \X,\y,g,\sigma^2)).
\label{eq:elbo}
\end{equation}
Here, $\DKL(q\, \|\, p)$ denotes the Kullback-Leibler (K-L) divergence
from a distribution $q$ to a distribution $p$
\citep{kullback1951information}.
% \begin{equation}
% F(q,g,\sigma^2) = 
% \mathbb{E}_q[\log p(\y \mid, \X, \b, \sigma^2)]
% - \DKL(q \,\|\, \pprior).
% \label{eqn:elbo_bayes_rule}
% \end{equation}
To aid in deriving the closed-form updates below, we reuse the
augmented-variable representation from the NM model, ${\bm \gamma}
\triangleq (\gamma_1, \ldots, \gamma_p)$, in which $\gamma_j$ was
defined in \eqref{eqn:latent}. The function $F$ is often called the
``evidence lower bound'' (ELBO) because it is a lower-bound for the
``evidence'', $\log p(\y \mid \X, g, \sigma^2)$.

% Although each of the terms on the right hand side of
% \eqref{prob:elbo_section2} are themselves intractable, some cancelling
% out occurs to make $F$ more tractable. (See
% Appendix~\ref{appendix:derivation}.) 

In \eqref{eq:EBsingle}, the optimization over $q$ is generally
intractable. The VEB approach addresses this by restricting the family
of distributions to be optimized. Specifically, our mean-field
VEB approach solves
\begin{equation} \label{prob:veb}
(\hat{q}, \hat{g}, \hat\sigma^2) = 
\argmax_{q \,\in\, \Q, \, 
g \,\in\, \mathcal{G}, 
\,\sigma^2 \,\in\, \mathbb{R}_{+}}
F(q,g,\sigma^2)
\end{equation}
where
\begin{equation} \label{def:mean-field}
\Q = \bigg\{q: q(\b,\bm\gamma) =
\prod_{j=1}^p q_j(b_j,\gamma_j) \bigg\},
\end{equation}
is the family of fully-factorized distributions.  The resulting
$\hat{q}$ factorizes into a product over individual factors
$\hat{q}_j(b_j, \gamma_j)$, $j = 1, \dots, p$, and serves as an
approximation to the EB posterior, $\hat{p}_\text{post}$. With this
constraint, the optimization becomes tractable (see
Section~\ref{subsec:coord} for details).

\subsubsection{Contrast with \cite{carbonetto2012scalable}}

While \citet{carbonetto2012scalable} also use a mean-field
approximation for linear regression, their approach to estimating
$g,\sigma^2$ is quite different---and substantially more
complex---than the VEB approach we describe here. In brief, they treat
$F(\qhat(g,\sigma^2),g,\sigma^2)$ as a direct approximation to the
evidence,
\begin{equation*}
p(\y \mid \X, g, \sigma^2) \approx
\exp\{F(\qhat(g,\sigma^2), g,\sigma^2)\},
\end{equation*}
and combine this with a prior distribution on $g, \sigma^2$ to arrive
at an approximate posterior distribution for $g, \sigma^2$.  This
approach is computationally burdensome because it requires finding a
separate approximation $\qhat$ for each $g, \sigma^2$. It also
requires specifying a prior on $(g,\sigma^2)$, which introduces an
additional layer of decision-making for the user. Our VEB approach
greatly simplifies this by simultaneously fitting $q, g,\sigma^2$ in a
single optimization \eqref{prob:veb}. Our approach also uses more
flexible prior families than \cite{carbonetto2012scalable}. And, as we
show in numerical experiments below, our approach generally improves
predictive performance.

\subsection{Coordinate Ascent Algorithm}
\label{subsec:coord}

We solve \eqref{prob:veb} for the multiple linear regression model
\eqref{model:bvs} with adaptive shrinkage priors
\eqref{eqn:normalmixture2} using a simple coordinate-wise approach,
outlined in Algorithm \ref{alg:coord}. While theoretical analysis
suggests that coordinate ascent can suffer poor rates of convergence
\citep{beck2013convergence, wright2015coordinate, hazimeh2018fast}, it
has the advantage of being guaranteed to converge to a stationary
point of the objective under mild conditions (see
Proposition~\ref{prop:conv} in Appendix \ref{appendix:proof}). In
practice, coordinate ascent has emerged as a simple, fast and reliable
approach for optimizing large-scale multiple linear regression models,
with both convex and nonconvex penalties
\citep{friedman2007pathwise,wu2008coordinate,friedman2010regularization,
  mazumder2011sparsenet, breheny2011coordinate, hazimeh2018fast}.

\begin{figure}[t]
\centering
\begin{minipage}{0.55\textwidth}
\begin{algorithm}[H]
\caption{Coordinate ascent for fitting \\
  VEB model (outline only).}
\label{alg:coord}
\begin{algorithmic}
\REQUIRE Data ${\bf X} \in \mathbb{R}^{n \times p}, {\bf y} \in
  \mathbb{R}^n$, \\ and initial estimates $q_1, \dots, q_p, g, 
  \sigma^2$.
\REPEAT
\FOR{$j \leftarrow 1 \mbox{ to } p$}
  \STATE $q_j \leftarrow \argmax_{q_j} F(q,g,\sigma^2)$
\ENDFOR
\STATE $g \leftarrow \argmax_{g \,\in\, \mathcal{G}} F(q,g,\sigma^2)$
\STATE $\sigma^2 \leftarrow
  \argmax_{\sigma^2 \,\in\, \mathbb{R}_{+}} F(q,g,\sigma^2)$
\UNTIL{termination criterion is met}
\RETURN
$q_1, \ldots, q_p, g, \sigma^2$.
\end{algorithmic}
\end{algorithm}
\end{minipage}
\end{figure}

In the following, we show that the steps in Algorithm~\ref{alg:coord}
are easy to implement:
\begin{enumerate}

\item[(i)] The update for each $q_j$ involves computing a posterior
  distribution under the NM model.

\item[(ii)] The update for $g$ involves running a single M-step update
  for the NM model, in which the exact posterior probabilities are
  replaced with approximate posterior probabilities.

\item[(iii)] The update for the residual variance, $\sigma^2$, has a
  simple, closed-form solution.

\end{enumerate}
These results are formally stated in the following proposition. (We
assume $\xj^T\xj = 1$ here to simplify the expressions; more general
results and algorithms that do not make this assumption are given
in Appendix~\ref{appendix:derivation}.)

\begin{proposition}[Coordinate ascent updates for VEB]
\label{prop:caisaderivation}
\rm Assume $\xj^T\xj = 1$, for $j = 1, \dots, p$, let $\bar{b}_j =
\mathbb{E}(b_j)$, $\bar{\b} = \mathbb{E}(\b)$ be the expected values
of $b_j, \b$ with respect to $q$, $\bar{\mathbf{r}} = \y - \X
\bar{\mathbf{b}} \in \mathbb{R}^n$ is the vector of expected residuals
with respect to $q$, $\X_{-j}$ denotes the design matrix $\X$
excluding the $j$th column, $q_{-j}$ is shorthand for all factors
$q_{j'}$, $j' \neq j$, and $\bar{\mathbf{r}}_j \in \mathbb{R}^n$ is
the vector of expected residuals accounting for linear effects of all
variables other than $j$,
\begin{equation}
\bar{\mathbf{r}}_j \triangleq \y - \X_{-j}\bar{\mathbf{b}}_{-j} 
= \y - \textstyle \sum_{j' \neq j} \mathbf{x}_{j'} \bar{b}_{j'}. 
\end{equation}
Additionally, we use $\tilde{b}_j \triangleq \xj^T \bar{\mathbf{r}}_j
= \bar{b}_j + \xj^T \bar{\mathbf{r}}$ to denote the ordinary least
squares (OLS) estimate of the coefficient $\bj$ when the residuals
$\rbar_j$ are regressed against $\xj$. Then we have the following
results:
\begin{enumerate}

\item[(i)] The coordinate ascent update $q_j^{\ast} \triangleq
  \argmax_{q_j} F(q, g, \sigma^2)$ is obtained by the posterior
  distribution for $\bj, \gamma_j$ under the normal means model
  (\ref{model:normal_means}, \ref{model:prior}) in which the
  observation $y_j$ is replaced by the OLS estimate of $\bj$; that is,
\begin{equation*}
q_j^{\ast}(b_j, \gamma_j = k) =
p_{\rm post}^{\rm NM}(b_j, \gamma_j = k; \tilde{b}_j, g, \sigma^2).
\end{equation*}
In particular, the posterior distribution at the maximum is
\begin{equation} 
\label{eqn:q_optimal_form}
q_j^{\ast}(b_j, \gamma_j = k) =
\phi_{jk}^{\ast} \, \N(b_j; \mu_{jk}^{\ast}, (s_{jk}^{2})^{\ast}),
\end{equation}
in which
\begin{align}
\mu_{jk}^{\ast} &= \mu_k(\tilde{b}_j; g, \sigma^2) \label{eqn:q_opt_mu} \\
(s_{jk}^{2})^{\ast} &= s_k^2(\tilde{b}_j; g, \sigma^2) \label{eqn:q_opt_s} \\
\phi_{jk}^{\ast} &= \phi_k(\tilde{b}_j; g, \sigma^2). \label{eqn:q_opt_phi}
\end{align}
See (\ref{eqn:NM-E}, \ref{def:posterior_mixture}, \ref{def:muk},
\ref{eqn:post_ebnm}) for the definitions of $\mu_k$, $s_k^2$, $\phi_k$
and $p_{\mathrm{post}}^{\mathrm{NM}}$.

\item[(ii)] The coordinate ascent update
\begin{equation*}
g^{\ast} \triangleq \argmax_{g \,\in\, \Gfm} F(q, g, \sigma^2)
\end{equation*}
is achieved by setting
\begin{equation}
\begin{aligned}
g^{\ast} &= \textstyle \sum_{k=1}^K \pi^{\ast}_k N(0,\sigma_k^2) \\
\pi_k^{\ast} &= \textstyle \frac{1}{p} 
\sum_{j=1}^p q_j(\gamma_j = k), \quad k = 1,\ldots,K.
\end{aligned}
\label{eqn:g_update}
\end{equation}
% And if $q_j$ is updated as in (i) above, meaning that $q_j(b_j) =
% p_{\rm post}^{\rm NM}(b_j \mid \tilde{b}_j, g, \sigma^2)$, then
% $q_j(\gamma_j = k)$ is equal to the ``responsibility''
% $\phi_{k}(\tilde{b}_j; g, \sigma^2)$.

\item[(iii)] Assuming $\sigma_1^2 = 0$, the coordinate ascent update
\begin{equation*}
(\sigma^{2})^{\ast} \triangleq
\argmax_{\sigma^2 \,\in \,\mathbb{R}_{+}} F(q,g,\sigma^2)
\end{equation*}
is achieved with
\begin{equation}
(\sigma^{2})^{\ast} = 
\frac{\|\rbar\|^2 + \sum_{j=1}^p \mathrm{Var}_q(b_j) +
      \sum_{j=1}^p \sum_{k=2}^K \phi_{jk} 
      \mathbb{E}[b_j \mid \gamma_j = k]/\sigma_k^2}
     {n + p - \sum_{j=1}^p \phi_{j1}}.
\label{eqn:sigma_update}
\end{equation}
Or, assuming $q_1, \ldots, q_p$ and $g$ have just been updated as in
(i) and (ii) above, we have the simpler update formula
\begin{equation}
(\sigma^{2})^{\ast} = \frac{\|\rbar\|^2 +
\bar{\bf b}^T (\tilde{\bf b} - \bar{\bf b}) +
\sigma^2 p(1 - \pi_1^{\ast})}
{n + p(1 - \pi_1^{\ast})}.
\label{eqn:sigma_update_simpler}
\end{equation}
% Note if ${\bm\pi}$ is not updated using \eqref{eqn:g_update},
% $\pi_1^{\ast}$ using can be substituted with $\sum_{j=1}^p
% \phi_1(\tilde{b}_j; g, \sigma^2)/p$ in
% \eqref{eqn:sigma_update_simpler}.
\end{enumerate}
\end{proposition}
\begin{proof}
See Appendix \ref{appendix:derivation}.
\end{proof}

\begin{algorithm}[t]
\caption{Coordinate ascent for fitting VEB model (in detail).}
\label{alg:caisageneral}
\begin{algorithmic}
\REQUIRE Data ${\bf X} \in \mathbb{R}^{n \times p}, {\bf y} \in
  \mathbb{R}^n$; number of mixture components, $K$; \\
  prior variances,
  $\sigma_1^2 < \cdots < \sigma_K^2$, with $\sigma_1^2 = 0$; initial
  estimates $\bar{\bf b}, {\bm\pi}, \sigma^2$.
\STATE $\bar{\mathbf{r}} = \y - \X\bar{\mathbf{b}}$
  \hfill (compute mean residuals)%
\STATE $t \leftarrow 0$
\REPEAT
\FOR{$j \leftarrow 1 \mbox{ to } p$}
    \STATE $\bar{\mathbf{r}}_j = \bar{\mathbf{r}} + \xj\bar{b}_j$
      \hfill (disregard $j$th effect in residuals)%
    \STATE $\tilde{b}_j \leftarrow \mathbf{x}_j^T\bar{\mathbf{r}}_j$.
      \hfill (compute OLS estimate)%
    \FOR{$k \leftarrow 1 \mbox{ to } K$} 
      \STATE $\phi_{jk} \leftarrow \phi_k(\tilde{b}_j; g, \sigma^2)$
      \STATE $\mu_{jk} \leftarrow \mu_k(\tilde{b}_j; g, \sigma^2)$
    \ENDFOR\hfill (update $q_j$; eqs.~\ref{eqn:q_opt_mu}, \ref{eqn:q_opt_phi})%
    \STATE $\bar{b}_j \leftarrow \sum_{k=1}^K \phi_{jk} \mu_{jk}$.
      \hfill (update posterior mean of $b_j$)%
    \STATE $\bar{\mathbf{r}} \leftarrow \bar{\mathbf{r}}_j - \xj\bar{b}_j$.
      \hfill (update mean residuals)%
  \ENDFOR
  \FOR{$k \leftarrow 1 \mbox{ to } K$}
    \STATE $\pi_k \leftarrow \sum_{j=1}^p \phi_{jk}/p$.
      \hfill (update $g$; eq.~\ref{eqn:g_update})
  \ENDFOR
  \STATE $\displaystyle \sigma^2 \leftarrow
    \frac{\|\bar{\mathbf{r}}\|^2 
      + \bar{\bf b}^T (\tilde{\bf b} - \bar{\bf b})
      + \sigma^2 p (1-\pi_1) }{n + p(1 - \pi_1)}$
    \hfill (update $\sigma^2$; eq.~\ref{eqn:sigma_update_simpler})%
  \STATE $t \leftarrow t + 1$
\UNTIL{termination criterion is met}
\RETURN $\bar{\bf b}, \pii, \sigma^2$
\end{algorithmic}
\end{algorithm}

Inserting these expressions into Algorithm \ref{alg:coord} (and
organizing computations to limit redundant operations and memory
requirements) yields Algorithm~\ref{alg:caisageneral}.  The
computational complexity of this algorithm is $O((n + K)p)$ per
outer-loop iteration, with memory requirements $O(n + p + K)$ (in
addition to storing the data matrix $\X$), making it tractable for
large data sets. The algorithm also exploits the fact that the
approximate posterior $q({\bf b}, {\bm \gamma})$ can be recovered from
$\pii, \sigma^2, \bar{\b}$ by running a single round of the
coordinate ascent updates for $q_1, \ldots, q_p$. Because of this, the
algorithm is initialized simply by providing an initial estimate of
$\bar{\bf b}$; the full $q$ is not needed. See Section
\ref{subsec:methods-initialization} for more on initialization.

\subsection{Accuracy of VEB and Exactness for Orthogonal Predictors}

\cite{carbonetto2012scalable} note that their variational
approximation approach provides the exact posterior distribution when
the columns of $\X$ are orthogonal.  Here we extend this result,
showing that in this special case the VEB method recovers the standard
EB method.
\begin{proposition} \label{prop:orthogonal}
\rm When $\X$ has orthogonal columns, the VEB approach
\eqref{prob:veb} is mathematically equivalent to the exact EB approach
(\ref{eqn:prior_estimation}, \ref{eqn:posterior_calculation}).
\end{proposition}
\begin{proof}
See Appendix~\ref{appendix:proof}.
\end{proof}
In brief, this result follows from the fact that, when $\X$ has
orthogonal columns, the (exact) posterior distribution for $\b$
factorizes as \eqref{def:mean-field}, and therefore the mean-field
assumption is not an approximation. By contrast, the ``conditional
maximum likelihood'' (CML) approach to approximating EB inference
\citep{george2000calibration, yuan2005efficient} is not exact even in
the case of orthogonal columns.

Proposition \ref{prop:orthogonal} suggests that our VEB method should
be accurate when the columns of $\X$ are close to orthogonal. It also
suggests that the approximation may be less accurate for very highly
correlated columns. However, \cite{carbonetto2012scalable} note that
even in this setting the estimated hyperparameters (here, $g$) can be
accurate. They also note that, when two predictors are highly
correlated and predictive of $\y$, the fully-factorized variational
approximation tends to give just one of them an appreciable estimated
coefficient. This is similar to the behavior of many PLR methods
including the Lasso (but different from a well-mixed MCMC-based
Bayesian method). While this behavior is undesirable when the main
aim of the analysis is to select variables for scientific
interpretation, it does not necessarily harm prediction
accuracy. Thus, the VEB approach may perform well for prediction even
in settings where the assumptions of the mean-field variational
approximation are clearly violated. Our numerical studies
(Section~\ref{sec:experiment}) confirm this.

\subsection{Practical Issues and Extensions}
\label{sec:practicalconsiderations}

In this subsection, we discuss some practical implementation issues
and potential extensions of this method.

\subsubsection{Intercept}

In multiple regression applications, it is common to include an
intercept term that is not regularized in the same way as other
variables. A common approach, and the approach we take here, is to
center $\y$ and the columns of $\X$ \citep{chipman2001practical,
  friedman2010regularization}; that is, the observed responses $y_i$
are replaced with the centered responses $y_i - \sum_{i'=1}^n y_{i'}/n$,
and the data $x_{ij}$ are replaced with column-centered values $x_{ij}
- \sum_{i'=1}^n x_{i'j}/n$.
 
\subsubsection{Selection of Grid for Prior Variances}
\label{subsec:grid}

Following \cite{stephens2016false}, we choose a grid $\{\sigma_1^2,
\ldots, \sigma_K^2\}$ that is sufficiently broad and dense so that
results do not change much if the grid is made broader and denser; the
aim is to choose a $\Gfm$ that closely approximates the non-parametric
family $\Gsmn$. Specifically, we set the lower end of the grid to be
$\sigma_1^2 = 0$, which is a point mass at zero, and we set the
largest prior variance to be $\sigma_K^2 \approx n$ so that the prior
variance of $\xj \bj$ is close to $\sigma^2$ (recall, we assumed
$\xj^T\xj = 1$, so $\mathrm{Var}(\xj) \approx 1/n$ when $\xj$ is
centered). We have found that 20 grid points spanning this range to be
good enough to achieve reliable prediction performance across many
settings (see Section~\ref{sec:experiment}). Based on this, our
default grid is the sequence $\sigma_k^2 = n(2^{(k-1)/K} - 1)^2$, $k =
1, \ldots, 20$. In rare cases, $\mathrm{Var}(\xj\bj)$ may be larger
than $\sigma^2$ for some $j$. In that case, we may need a larger
$\sigma_K^2$ to avoid underestimating, or ``overshrinking'', the
effect $\bj$. Therefore, we suggest checking that the final estimate
of $\pi_K$ is negligible and, if not, the grid can be made wider.

\subsubsection{Initialization and Update Order}
\label{subsec:methods-initialization}

Except in special cases, maximizing $F$ is a nonconvex optimization
problem, and so although Algorithm \ref{alg:caisageneral} is
guaranteed to converge, the solution obtained may depend on
initialization of $\bbar$, $\pii$, $\sigma^2$, as well as the order in
which the coordinate ascent updates cycle through the coordinates $j
\in \{1, \ldots, p\}$ \citep[e.g.,][]{carbonetto2012scalable,
  ray2019variational}. Therefore, we experimented with different
initialization and update orderings.

The simplest initialization for $\bbar$ is the ``null initialization''
$\bbar=(0, \ldots, 0)^T$.  We also consider initializing $\bbar$ to
the Lasso solution $\hat{\b}^{\mathrm{lasso}}$ in which the Lasso
penalty is chosen via cross-validation.  Given $\bbar$, we initialize
the residual variance as $\sigma^2 = \|\y - \X\bbar\|^2/n$, and we
initialize the mixture weights to $\pii = (1/K, \ldots, 1/K)$. In
numerical experiments (Appendix \ref{subsec:initorder}) we found that
the Lasso initialization often improved predictive performance when
columns of $\X$ were highly correlated. In other cases, the null and
the Lasso initializations performed similarly. With the Lasso
initialization, we did not find any systematic benefit to different
update orderings. Therefore, our default approach is to use the Lasso
initialization with updates performed in the natural order,
$1,2,\dots,p$.

\subsubsection{Termination Criterion}

We stop iterating when estimates of the prior distribution $g$
stabilize; specifically, we stop at iteration $t$ if
$\norm{{\bm\pi}^{(t)} - {\bm\pi}^{(t-1)}}_{\infty} < K \times
10^{-8}$.

\subsubsection{Computing the ELBO}

Although Algorithm \ref{alg:caisageneral} optimizes the ELBO, $F$, it
does not require computation of the ELBO. Still, it can be useful to
compute the ELBO, for example to monitor progress of the coordinate
ascent updates, or to compare fits obtained from different runs of the
algorithm. Appendix \ref{appendix:derivation} includes analytic
expressions for the ELBO that may be useful in such cases.
%
% (equations \ref{eqn:elbo_parametrized}, \ref{eqn:expected_log_lik},
% \ref{eqn:neg_kl_div}).
%
% But all the expressions needed to compute these quantities are given
% in Sections \ref{sec:background} and \ref{sec:mrashmodel}.

\subsubsection{Extension to Other Mixture Prior Families}

We have focussed on the normal mixture prior family $\G = \Gfm$
because it makes computations simple, fast, and numerically
stable. Furthermore, this prior family includes most prior
distributions previously used for multiple regression, and so we
expect it to suffice for many practical applications. However,
Algorithm \ref{alg:caisageneral} could be adapted to accommodate other
prior families of the form $\mathcal{G} = \{g = \sum_{k=1}^K \pk g_k :
\pii \in \mathbb{S}^K\}$, with fixed mixture components $g_1, \ldots,
g_K$. When choosing $\mathcal{G}$, there are two important practical
points to consider: (i) the convolution of $g_k$ with a normal
likelihood should be numerically tractable, ideally with an analytic
expression; (ii) the posterior mean in the normal means model with
prior $\bj \sim g_k$ should be easy to compute. Examples of fixed
mixture components $g_k$ satisfying (i) and (ii) include point masses,
uniform distributions and Laplace distributions.

\subsubsection{Inference and Variable Selection}

We have focussed on developing flexible multiple regression methods
for accurate {\em prediction}, which requires only a point estimate
for $\b$ (e.g., the posterior mean, $\bar\b$).  However, the
approximate posterior distributions from our method could also be used
for inference, that is, to assess uncertainty in the estimated
$\b$. For example, to assess significance of each variable in the
regression, it is easy to compute the local false sign rate,
\textit{lfsr} \citep{stephens2016false}, which quantifies confidence
in the sign of the effect (and which we generally prefer to the
closely related {\em local false discovery rate}; see
\citealt{stephens2016false}). However, caution is warranted in
settings with highly correlated variables: in such settings, the
approximate posterior distribution from the fully-factored variational
approximation will often be inaccurate. While predictive performance
is quite robust to this issue, inference is more sensitive
\citep{carbonetto2012scalable}.  Thus, other methods may be preferred
for inference with highly correlated variables; see
\cite{wang2018simple} for further discussion.

\section{Connecting VEB and Penalized Linear Regression}
\label{sec:penalty}

This section shows that the approximate posterior mean computed by our
VEB approach also solves a PLR with a nonconvex penalty, where the
form of the penalty is flexible and is automatically learned from the
data without need for cross-validation.

\subsection{Penalties and Shrinkage Operators}

\begin{figure}[!t]
\centering
\includegraphics[width=0.95\textwidth]{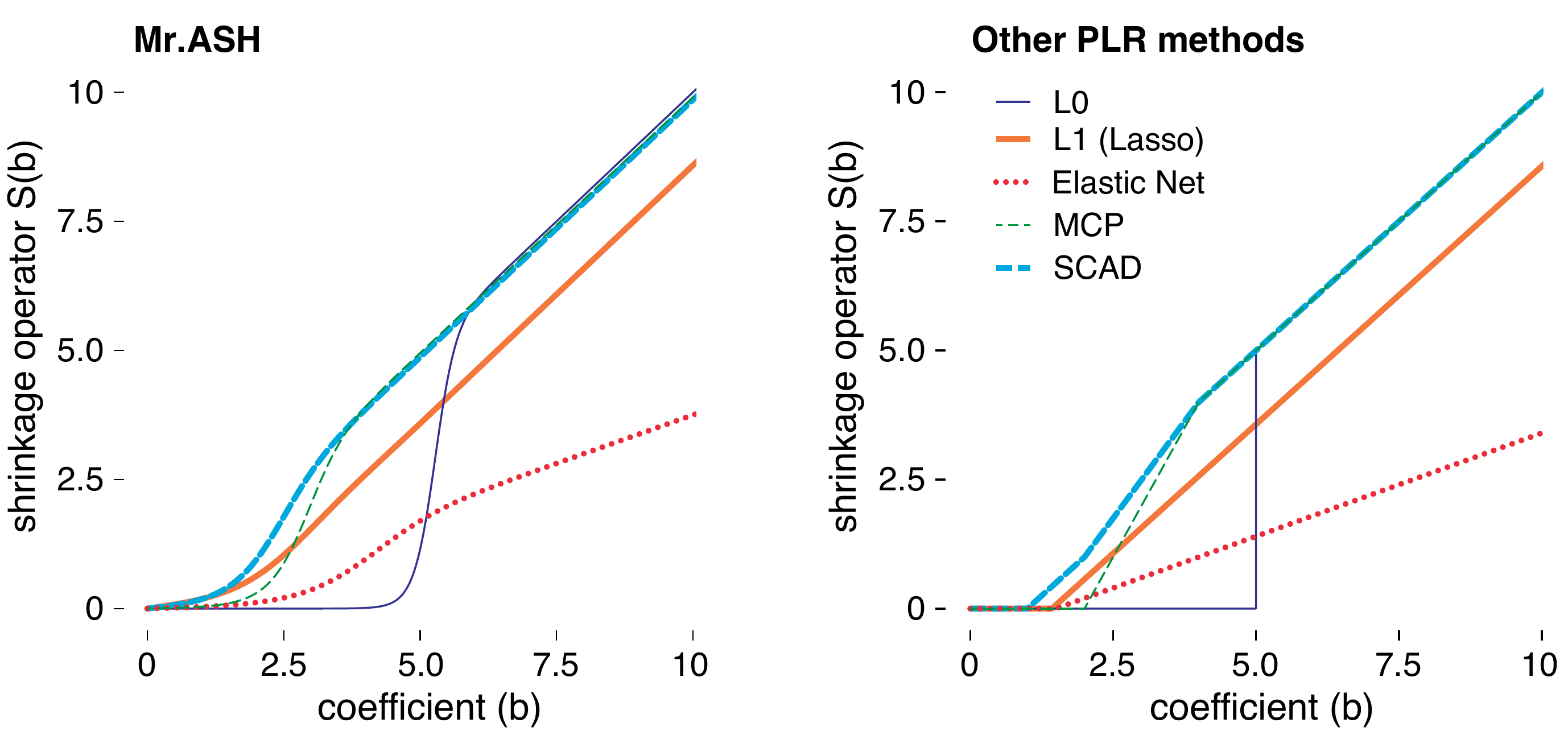}
\caption{Examples of posterior mean shrinkage operators for different
  $g \in \Gfm$ (left-hand panel) and $\sigma^2$ that were chosen to
  mimic the shrinkage operators from some commonly used penalties
  (right-hand panel).}
\label{fig:flexibility}
\end{figure}

Penalized linear regression (PLR) methods estimate the regression
coefficients by minimizing a penalized squared-loss function:
\begin{equation}
\label{eqn:penalized_linear_regression}
\underset{\b \,\in\, \mathbb{R}^p}{\mathrm{minimize}} \;
h_{\rho}(\b) \triangleq
\frac{1}{2}\norm{\y - \X\b}^2 + \sum_{j=1}^p \rho(\bj),
\end{equation}
for some penalty function $\rho : \mathbb{R} \to \mathbb{R}$. As
mentioned above, the PLR problem
\eqref{eqn:penalized_linear_regression} is often tackled using
coordinate descent algorithms;
%
% \cite[e.g.,][]{friedman2007pathwise, friedman2010regularization,
% wu2008coordinate, breheny2011coordinate, hazimeh2018fast}
%
that is, by iterating over the coordinates of $b_1, \ldots, b_p$
sequentially, at each iteration solving
\eqref{eqn:penalized_linear_regression} for one coordinate $b_j$ while
keeping the remaining coordinates fixed. Assuming the columns of $\X$
are scaled such that $\xj^T\xj = 1$, the solution for the $j$th
coordinate is obtained by
\begin{equation}
\label{eqn:coord_descent}
\bj \leftarrow S_{\rho}(\bj + \xj^T(\y - \X\b)),
\end{equation}
where 
\begin{equation}
\label{def:shrinkage_thresholding_operator}
S_{\rho}(t) \triangleq
\argmin_{\theta \,\in\, \mathbb{R}}\frac{1}{2}(t - \theta)^2 
+ \rho(\theta)
\end{equation}
is the ``shrinkage operator'' (or ``univariate proximal operator'';
\citealt{parikh2014proximal}) for penalty $\rho$. Studying the
shrinkage operator $S_{\rho}$ is often helpful for understanding the
behaviour of its corresponding penalty, $\rho$. For example, the
shrinkage operators for the $L_0$ and $L_1$ penalties both have a
``thresholding'' property in which coefficient estimates $t$ less than
some value are driven to zero. This thresholding property tends to
produce sparse solutions to \eqref{eqn:penalized_linear_regression};
see Figure~\ref{fig:flexibility} for an illustration. Table
\ref{table:penalty} in Appendix \ref{appendix:penalties} gives some
commonly used penalty functions and their corresponding shrinkage
operators.
%
% Our setup here reduces to the analysis of shrinkage
% operators in \cite{fu1998penalized} in the special case of an
% orthonormal matrix ${\bf X}$, and in this special case $\bj + \xj^T(\y
% - \X\b)$ in \eqref{def:shrinkage_thresholding_operator} simplifies to
% the ordinary least squares (OLS) estimate of $b_j$,
% $\hat{b}_j^{\mathrm{OLS}} \triangleq {\bf x}_j^T{\bf y}$, and
% $S_{\rho}(\hat{b}_j^{\mathrm{OLS}})$ is therefore the shrunken OLS
% estimate.

% Some commonly used penalty functions include: the $L_2$-penalty
% (Ridge; \citealt{hoerl1970ridge}); the $L_0$-penalty
% \citep{miller2002subset}; the $L_1$-penalty (Lasso;
% \citealt{tibshirani1996regression}); the Bridge penalty which includes
% the Ridge and Lasso as special cases \citep{fu1998penalized}; the
% Elastic Net, a mixture of $L_1$ and $L_2$ penalties
% \citep{zou2005regularization}; the Smoothly Clipped Absolute Deviation
% (SCAD; \citealt{fan2001variable}); and the Minimax Concave Penalty
% (MCP; \citealt{zhang2010nearly}). 

In the next section, we show that our VEB approach can be interpreted
as solving a PLR problem with a penalty function $\rho_{g_{\sigma},
  \sigma}$ that depends on the prior $g_\sigma$ and the residual
variance $\sigma^2$. This penalty function has a corresponding
shrinkage operator, denoted by $S_{g_{\sigma}, \sigma}$, that has a
particularly simple form and interpretation: {\em it is the posterior
  mean under a normal means model}. It is formally defined as follows.
\begin{definition}[Normal Means Posterior Mean Operator]
\label{def:posterior_mean_shrinkage}
\rm Define the {\em normal means posterior mean operator,}
$\shrinkf: \mathbb{R} \to \mathbb{R}$, as the mapping
\begin{equation}
\label{eqn:shrink}
\shrinkf(y) \triangleq \mathbb{E}_{\rm NM}(b \mid y, f,\sigma^2),
\end{equation}
where $\mathbb{E}_{\mathrm{NM}}$ denotes the posterior expectation
under the following normal means model with prior $f$ and variance
$\sigma^2$,
\begin{equation}
\label{eqn:nmf}
\begin{aligned}
y \mid b, \sigma^2 &\sim N(b, \sigma^2) \\
b &\sim f.
\end{aligned}
\end{equation}
\end{definition}
From \eqref{eqn:post_ebnm}, $\shrinkf$ has a simple analytic form
when $f = g_\sigma$, $g \in \Gfm$:
\begin{equation}
\label{eqn:shrink-gmf}
\shrinkg(y) = \sum_{k=1}^K \phi_k(y; g, \sigma^2) \, \mu_k(y; g, \sigma^2).
\end{equation}
It is easy to show that $\shrinkg$ is an odd function and is monotonic
in $y$. Also, $\shrinkg$ is a shrinkage operator, in that
$|\shrinkg(y)| \leq |y|$; see Lemma~\ref{lem:property_pm_shrinkage} in
Appendix \ref{appendix:proof}. Indeed, given a suitable choice of
prior family, the shrinkage operator $\shrinkg$ can qualitatively
mimic the behavior of many commonly used shrinkage operators (Figure
\ref{fig:flexibility}).

The behavior of $\shrinkg$ naturally depends on the prior. For
example, the more mass the prior places near zero, the stronger the
shrinkage toward zero. Our VEB method estimates the prior from within
a flexible family capable of capturing a wide range of scenarios;
consequently, the corresponding shrinkage operator is also estimated
from a flexible family of shrinkage operators. This process is
analogous to estimating the tuning parameters in regular PLR methods
which is usually done by cross-validation (CV).
%
%\citep[e.g.,][]{friedman2010regularization,
%  breheny2011coordinate}
%
However, our VEB approach dispenses with CV and makes it possible to
efficiently tune across a much wider range of shrinkage
operators.

% The shrinkage operator $\shrinkg$ describes the shrinkage effect for a
% given choice of prior on $b$. Drawing a connection between penalties
% and Bayesian priors is not new. But previous treatments have brought
% the perspective of PLR as maximizing a posterior distributon with a
% particular choice of prior on ${\bf b}$ (e.g.,
% \citealt{fu1998penalized, Figueiredo2003}). By contrast, VEB as PLR is
% based on an averaging (mean) operation instead of a maximization
% (eq.~\ref{eqn:shrink}).
%
% Further, in practice the computational burden of the VEB approach is
% similar to methods that tune one parameter by CV, and is much lower
% than methods that tune two parameters by CV (see Section
% \ref{subsec:predperformance}).

\subsection{VEB as Penalized Linear Regression}

% Having provided some intuition for the connection between our VEB
% approach and PLR, now we develop a more formal connnection between the
% two. An obvious difference between the PLR problem
% \eqref{eqn:penalized_linear_regression} and our VEB approach is that
% the PLR directly estimates the regresssion coefficients $\b \in
% \mathbb{R}^p$ by optimizing an objective function, which we will
% denote by $h_{\rho}(\b)$, whereas our VEB method estimates the
% regression coefficients indirectly: it first estimates the approximate
% posterior distributions $q_1, \ldots, q_p$ by optimizing an objective
% function (the ELBO, $F$), then it estimates the regresssion
% coefficients as $\b = \bbar$, in which $\bbar \triangleq (\bar{b}_1,
% \ldots, \bar{b}_p)^T$, $\bar{b}_j \triangleq \mathbb{E}_q[b_j]$.
 
The first step to connecting our VEB approach to a PLR is to write it
as an optimization over the regression coefficients, $\b$, rather than
an optimization over the (approximate) posterior distributions,
$q$. To do this, we define an objective function:
\begin{equation}
\label{eqn:define_elbo_as_plr}
h(\bbar, g, \sigma^2) \triangleq
-\bigg\{\max_{q \,\in\, \mathcal{Q}, \, \mathbb{E}_q[\b] \,=\, \bbar}
\; F(q,g,\sigma^2) \bigg\}.
\end{equation}
The constraint $\mathbb{E}_q[\b] \,=\, \bbar$ means that the expected
value of $\b$ with respect to $q$ is $\bbar$. The negative sign is
introduced to align with the convention that PLRs are usually
minimization problems, as in \eqref{eqn:penalized_linear_regression}.

Any algorithm for optimizing $F$ over $q$ (and possibly $g, \sigma^2$)
also provides a way to optimize $h$ over $\bbar$ (and possibly
$g, \sigma^2$), as formalized in the following proposition.
\begin{proposition}[Computing Posterior Mean as an Optimization Problem]
\label{prop:vpm_solves_an_opt}
\rm Let $\qhat$, $\ghat$, $\shat$ be a solution to
\begin{equation*}
\qhat, \ghat, \shat =
\argmax_{q \,\in\, \Q,\, g \,\in\, \G,\, \sigma^2 \,\in\, \mathcal{T}} \,
F(q,g,\sigma^2),
\end{equation*}
where $\Q$ is the variational mean-field family of approximate
posterior distributions \eqref{def:mean-field}, $\G$ is any family of
prior distributions on $b \in \mathbb{R}$, and $\mathcal{T}$ is any
subset of $\mathbb{R}_{+}$. (This general formulation allows, as a
special case, $g,\sigma^2$ to be fixed by taking $\G$ and
$\mathcal{T}$ to be singleton sets.) Let $\hat{\b}$ denote the
expected value of $\b$ with respect to $\qhat$. Then $\hat{\b}, \ghat,
\shat$ also solves the following optimization problem:
\begin{equation*}
\hat{\b}, \ghat, \shat =
\argmin_{\bbar \,\in\, \mathbb{R}^p,\, g \,\in\, \G,
\,\sigma^2 \,\in\, \mathcal{T}}
\; h(\bbar, g, \sigma^2).
\end{equation*}
\end{proposition}
\begin{proof}
See Appendix~\ref{appendix:proof}.
\end{proof}

The final step to connecting VEB with PLR is to show that $h$ has the
form of a penalized squared-loss function.
\begin{theorem}[VEB as a Penalized Log-likelihood]
\label{thm:veb_penloglik}
\rm The objective function $h$ defined in
\eqref{eqn:define_elbo_as_plr} has the form of a PLR,
\begin{equation}
h(\bbar, g, \sigma^2) =
\frac{1}{2\sigma^2} \norm{\y - \X\bbar}^2
+ \frac{1}{\sigma^2} \sum_{j=1}^p \rho_{g_{\sigma}, \sigma}(\bar{b}_j)
+ \frac{n-p}{2} \log(2\pi\sigma^2),
\label{eqn:elbo_as_plr_full_expression}
\end{equation}
in which the penalty function $\rho_{f,\sigma}$ satisfies
\begin{equation}
\label{eqn:penalty_main_text}
\rho_{f,\sigma}(\shrinkf(y)) =
- \sigma^2 \ell_{\mathrm{NM}} (y; f, \sigma^2)
- \textstyle \frac{1}{2}(y - \shrinkf(y))^2,
\end{equation}
and
\begin{equation}
\rho_{f,\sigma}'(\shrinkf(y)) = (y - \shrinkf(y)).
\label{eqn:penderiv}
\end{equation}
Here, $\ell_{\rm NM}(y; f, \sigma^2) \triangleq \log p(y \mid f,
\sigma^2)$ denotes the marginal log-likelihood under the NM model
\eqref{eqn:nmf}, and $\shrinkf$ denotes the shrinkage operator
\eqref{eqn:shrink}.
\end{theorem}
\begin{proof}
See Appendix~\ref{appendix:proof}.
\end{proof}

From this theorem, it follows that the NM posterior mean shrinkage
operator $\shrinkf$ \eqref{eqn:shrink} can be also written in the form
of \eqref{def:shrinkage_thresholding_operator}, a shrinkage operator
for the penalty $\rho_{f,\sigma}$.  Explicit computation of
$\rho_{f,\sigma}(\bar{b})$ for a given $\bar{b}$ in
\eqref{eqn:penalty_main_text} would require computing the inverse
shrinkage operator, $y = S_{f,\sigma}^{-1}(\bar{b})$. This inverse
exists because the shrinkage operator \eqref{eqn:shrink-gmf} is
strictly increasing; however, we do not have an analytic expression
for the inverse, so we do not have an analytic expression for
$\rho_{f,\sigma}(\bar{b})$.

\begin{figure}[t]
\centering
\begin{minipage}{0.5\textwidth}
\begin{algorithm}[H]
\caption{Coordinate Ascent Iterative \\ Shrinkage 
   Algorithm for Variational Posterior \\ Mean (with fixed $g,\sigma^2$)}
\label{alg:caisa-simple}
\begin{algorithmic}
\REQUIRE ${\bf X} \in \mathbb{R}^{n \times p}, {\bf y} \in
  \mathbb{R}^n$, $\sigma^2 > 0$, prior $g$, and initial estimates $\bbar$.
\REPEAT
\FOR{$j \leftarrow 1 \mbox{ to } p$}
\STATE $\bar{b}_j \leftarrow \shrinkg(\bar{b}_j + \xj^T(\y - \X\bbar))$
\ENDFOR
\UNTIL{convergence criteria is met}
\RETURN $\bbar$
\end{algorithmic}
\end{algorithm}
\end{minipage}
\end{figure}

\subsubsection{Special Case When $g$ and $\sigma^2$ Are Fixed}

The special case when $g$ and $\sigma^2$ are fixed is particularly
simple and helpful for intuition. In this case, the VEB approach is
solving a PLR problem with fixed penalty $\rho_{{g_{\sigma}},\sigma}$
and shrinkage operator $S_{g_{\sigma},\sigma}$. This leads to a simple
coordinate ascent algorithm (Algorithm~\ref{alg:caisa-simple}).
Compare this with the inner loop of Algorithm \ref{alg:caisageneral}
which maximizes the ELBO, $F$, over each $q_j$ in turn. The key
computation in the inner loop is the computation of the posterior
mean, $\bar{b}_j$.  (When $g$ and $\sigma^2$ are fixed, computing
$\phi_{jk}$ and $\mu_{jk}$ is needed only to compute $\bar{b}_j$.)
Further, by Proposition \ref{prop:caisaderivation}, this value is
computed as the posterior mean under a simple NM model, which is given
by the shrinkage operator $\shrinkg$.

In summary, for fixed $g, \sigma^2$, Algorithm \ref{alg:caisageneral}
can be reframed as a coordinate ascent algorithm for PLR, which is
Algorithm \ref{alg:caisa-simple}.

\subsubsection{Special Case of a Normal Prior (Ridge Regression)}

When the prior, $g$, is a fixed normal distribution with zero mean,
the NM posterior mean shrinkage operator $\shrinkg$ is the same as the
ridge regression (or $L_2$) shrinkage operator (Table
\ref{table:penalty}) and the penalty function $\rho_{{g_{\sigma}},
  \sigma}$ is the $L_2$-penalty. Thus, in this special case, Algorithm
\ref{alg:caisa-simple} is solving ridge regression ({\em i.e.}, PLR
with $L_2$-penalty), which is a convex optimization
problem. Furthermore, in this special case Algorithm
\ref{alg:caisa-simple} converges to the {\em exact} posterior mean of
$\b$ because the posterior is multivariate normal, and therefore the
posterior mean is equal to the posterior mode. Thus, in this special
case, even though the variational posterior approximation $q$ does not
exactly match the true posterior---the true posterior does not
factorize as in \eqref{def:mean-field}---the variational posterior
mean recovers the true posterior mean.

\subsubsection{Posterior Mean vs. Posterior Mode}

Traditional PLR approaches are sometimes motivated from a Bayesian
perspective as computing a posterior mode estimate for $\b$---{\em
  i.e.}, a {\em maximum a posteriori} (MAP) estimate---in which the
penalty term corresponds to some prior on $\b$
\citep{fu1998penalized}. For example, the Lasso is the MAP with a
Laplace (``double-exponential'') prior \citep{Figueiredo2003,
  park2008bayesian, tibshirani1996regression}. By contrast, the
variational approach seeks the {\em posterior mean}, not the posterior
mode, and likewise the VEB shrinkage \eqref{eqn:shrink} is based on an
averaging (mean) operation instead of the usual maximization (mode).
Our formulation of the (approximate) posterior mean as solving a PLR
is new, at least as far as we are aware, and this formulation may be
useful in other settings.

% For example, the variational approach with a Laplace prior would not
% lead to the usual Lasso estimate---it would be closer to the posterior
% mean from the Bayesian Lasso \citep{park2008bayesian}.

From a Bayesian decision-theoretic perspective (e.g., Chapter 4 of
\citealt{berger2013statistical}), the posterior mean for $\b$ has
better theoretical support than the posterior mode; not only does it
minimize the expected mean-squared error in $\b$, it also minimizes
the expected mean-squared error in the predicted, or ``fitted,''
responses $\hat{y}_i = (x_{i1}, \ldots, x_{ip})^T \b$. Although the
posterior mode can have attractive properties---for example, the
posterior mode with a Laplace prior is sparse whereas the posterior
mean with a Laplace prior is not---the posterior mode has very little
support as an estimator for $\b$, particularly when predicting $\y$ is
the main goal. To illustrate this, consider a spike-and-slab prior
\citep{mitchell1988bayesian} with non-zero mass at zero: the posterior
mode is always $\b = {\bf 0}$, which will generally provide poor
prediction performance. Also, consider that if
$\hat{\b}^{\mathrm{mode}}$ is the posterior mode of $\b$, then
$(x_{i1}, \ldots, x_{ip})^T \hat{\b}^{\mathrm{mode}}$ is not generally
the posterior mode of the fitted value $\hat{y}_i$.

\section{Numerical Experiments}
\label{sec:experiment}

We used simulations to empirically assess the predictive performance
of our proposed method and compare it with other methods. We call
the proposed method {\em multiple regression with adaptive shrinkage
  priors,} or ``\mrash''. \mrash{} is implemented in the R package {\tt
  mr.ash.alpha}, available at
\url{https://github.com/stephenslab/mr.ash.alpha}.  In the experiments,
 we used {\tt mr.ash.alpha} version 0.1-33 (git commit id
0845778). The source code and analysis steps used to generate the
results of our numerical experiments are included in a separate
repository on GitHub,
\url{https://github.com/stephenslab/mr-ash-workflow}.

\subsection{Methods Compared}
\label{subsec:methods-compared}

We compared \mrash{} with 12 different methods (Table
\ref{table:listofcompetitors}). These include a range of well
established and recently developed methods based on both PLR and
Bayesian ideas. These methods vary in (1) the choice of penalty or
prior (and possibly other modeling assumptions); and (2) the algorithm
used to fit the model (e.g., point estimation of ${\bf b}$
vs. approximate posterior inference of ${\bf b}$ via MCMC or
variational inference). When selecting methods to compare, we tried to
choose methods that were publicly available as open source software
and that were well maintained and documented. Since \mrash{} and many
other popular large-scale linear regression methods are implemented in
R, we sought out methods implemented in R \citep{R}.

Here we give a brief overview of the different methods and point out
some of their expected strengths and weaknesses:
\begin{itemize}

\item Ridge regression (\citealt{hoerl1970ridge}) and the Bayesian
  Lasso (\citealt{park2008bayesian, bglr}) are
  well adapted to dense signals, so should be competitive in such
  settings. On the other hand, they may perform poorly for sparse
  signals.

\item \lzero{} \citep{hazimeh2018fast} and \susie{} \citep{wang2018simple}
  are better adapted to sparse signals, so they should perform well in
  such settings. On the other hand, their assumptions are poorly
  suited to more dense signals.

\item The \lasso{} \citep{tibshirani1996regression} is one of the most
  widely used PLR methods. Computing the \lasso{} estimator is a
  convex optimization problem. One of the well-studied issues is that
  \lasso{} estimates can suffer from bias by overshrinking the
  strongest signals \cite[e.g.,][]{su2017false,
    javanmard2018debiasing}.

\item The Elastic Net (\citealt{zou2005regularization}) is another
  widely used convex PLR method. It has two tuning parameters. The
  Elastic Net penalty is more flexible than the \lasso{} and ridge
  regression and, indeed, it includes both as special cases. The
  Elastic Net may therefore perform well across a wider range of
  settings, at the cost of increased computation in the parameter
  tuning.

\item \scad, \mcp, the Spike-and-Slab Lasso (SSLasso) and the
  Trimmed Lasso are methods based on nonconvex or adaptive penalties
  that were designed, in part, to address limitations of the \lasso{}
  penalty \citep{bai2021spike, bertsimas2017trimmed,
    breheny2011coordinate, rovckova2018spike, trimmedlasso,
    yun2019trimming}. They might therefore outperform the \lasso{} (and
  the Elastic Net), possibly at the cost of some additional
  computation. Since these methods were primarily developed with
  sparse regression in mind, they may not always perform as well for
  dense signals.

\item \bayesb{} is a Bayesian regression method with a
  ``spike-and-slab'' prior, in which the ``slab'' is a $t$
  distribution \citep{meuwissen2001prediction, bglr}. It has the
  potential to perform well for both sparse and dense signals.  One
  concern is that the Markov chain may or may not be simulated long
  enough so as to adequately explore the posterior distribution.

\item \varbvs{} \citep{carbonetto2012scalable, varbvs} and \mrash{} both
  compute approximate posteriors using the same mean field variational
  approximation. Compared to \varbvs, \mrash{} features a more flexible
  prior, and uses a simpler and more efficient empirical Bayes
  approach to estimate the prior. \mrash{} also uses an initialization
  based on the \lasso, which, as we show below, can improve the model
  fit, particularly when the predictors are strongly correlated, or
  correlated in complex ways. We expect \mrash{} to outperform \varbvs{}
  some settings, particularly in ``dense'' settings when many of the
  predictors affect the outcome, or when the predictors are strongly
  correlated.

\end{itemize}

As an additional point of comparison in ``sparse'' data simulations,
we also show results for the ``Oracle OLS'' method, which is the
ordinary least-squares (OLS) estimate of $\b$ conditional on knowledge
of which coefficients $b_j$ are non-zero. This can be considered a
lower bound on the achievable prediction error when the number of
non-zero coefficients is small. (When the number of non-zero
coefficients is large, the Oracle OLS will perform poorly, so we do
not include the Oracle OLS result in settings with many non-zero
coefficients.)
  
A factor that inevitably complicates comparisons is that most methods
have many options and tuning parameters. Even a relatively
straightforward method such as the \lasso{} has multiple tuning
parameters that can affect performance, some of which involve
tradeoffs in computing effort versus prediction accuracy: number of
folds to use in the K-fold cross-validation step; what criterion to
use for selecting the optimal penalty strength parameter; whether to
``relax'' the fit; etc. For each method, we tried to follow the
recommendations given in the software documentation or in published
papers, and in some cases we changed the default settings to improve
performance. However, with so many methods to compare, it was
infeasible to find the best settings for each
method. Appendix~\ref{appendix:methods_compared_details} includes
additional information on how the methods were run in these
experiments.

\begin{table}[t]
\centering
\begin{tabular}{lcM{3.5in}}
method &
R package &
brief description \\
\midrule
\multicolumn{3}{l}{\bf PLR methods} \\
ridge regression & glmnet & PLR with convex $L_2$ penalty \\
Lasso & glmnet & PLR with convex $L_1$ penalty \\
Elastic Net & glmnet & PLR with linear combination of 
$L_1$ and $L_2$ penalties \\ 
SCAD & ncvreg & PLR with nonconvex SCAD penalty \\
MCP & ncvreg & PLR with minimax concave penalty \\ 
L0Learn & L0Learn & PLR with nonconvex $L_0$, $L_0L_1$ or $L_0L_2$ penalty \\
SSLasso & SSLASSO & 
PLR with adaptive penalty based on a Laplace mixture prior \\
Trimmed Lasso &
(none) &
PLR with nonseparable ``trimmed lasso'' penalty
\\[1em]
\multicolumn{3}{l}{\bf Bayesian and empirical Bayes methods} \\
 BayesB & BGLR & MCMC with spike-and-slab prior (the ``slab'' is $t$) \\
Bayesian Lasso & BGLR & MCMC with scaled Laplace prior \\ 
varbvs & varbvs & variational inference with
spike-and-slab prior \\
SuSiE & susieR & variational inference for ``SuSiE'' model
\end{tabular}
\caption{Summary of methods compared in the simulations. All methods
  are implemented in R packages except the Trimmed Lasso which is
  implemented in MATLAB.}
\label{table:listofcompetitors}
\end{table}

\subsection{Design of Simulations}
\label{subsec:settings}

To test the methods in a wide variety of settings, we designed five
sets of simulations, which we refer to as ``Experiment 1'' through
``Experiment 5.'' In each set of simulations, we varied one aspect of
the simulation while keeping the other aspects fixed.
\begin{itemize}
  
\item In Experiment 1, we varied the ``sparsity level''; that is, the
  proportion of variables with non-zero coefficients. We denote the
  sparsity level by $s$, so that $s = 1$ is the sparsest model (with
  only a single non-zero coefficient) and $s = p$ is the densest model
  (all variables affect $\y$).
  
\item In Experiment 2, we varied the ``total signal strength''; 
  specifically, the proportion of variance in the response $\y$ that is
  explained by $\X$. We refer to this parameter as ``PVE'', short for
  ``proportion of variance explained.''
  
\item In Experiment 3, we considered different distributions for the
  non-zero coefficients. We use $h$ to denote the distribution that was
  used to simulate the non-zero coefficients.
  
\item In Experiment 4, we varied the number of predictors, $p$.

\item In Experiment 5, we simulated residual errors (noise) in
  different ways. This was done to assess how departures from the
  assumption of normally-distributed noise---an assumption made by
  most methods---might affect performance.
  
\end{itemize}
These experiments focus on the case $p > n$. However, we
  note that \mrash{} also performs well in the easier case where
  $p \ll n$; see Appendix \ref{subsec:p<n}.

In all the experiments, we took the following steps to simulate each
data set:
\begin{itemize}

\item First, we generated the $n \times p$ design matrix, ${\bf
  X}$. We considered three types of design matrices: (1) {\em
  independent variables}, in which the individual observations
  $x_{ij}$ were simulated {\em i.i.d.} from the standard normal
  distribution; (2) {\em correlated variables}, in which each row of
  ${\bf X}$ was an independent draw from the multivariate normal
  distribution with mean zero and a covariance matrix diagonal entries
  set to 1 and off-diagonal entries set to $\rho \in [0,1]$; and (3)
  {\em real genotype data}, in which $\X$ was a genotype data matrix
  from the Genotype-Tissue Expression (GTEx) project
  \citep{gtex2017genetic}. (Specifically, we used the processed
  genotype data sets generated in \citealt{wang2018simple}.) In these
  data sets, the variables were genetic variants---specifically,
  single nucleotide polymorphisms, or ``SNPs''---and the SNPs
  exhibited complex correlation patterns.  Some SNP pairs had very
  strong correlations, approaching 1 or $-1$. Each genotype matrix
  $\X$ contained the genotypes of all SNPs within 1 Megabase (Mb) of a
  gene's transcription start site after filtering out SNPs with minor
  allele frequencies less than 5\% (see \citealt{wang2018simple} for
  details). Among the thousands of data sets used in
  \citealt{wang2018simple}, we randomly selected 20 data sets for our
  simulations. Unless otherwise stated, we simulated independent
  variables with $n = 500$ and $p = \mbox{1,000}$. For the genotype
  data sets, $n = 287$, and $p$ ranged from $\mbox{4,012}$ to
  $\mbox{8,760}$. Each genotype matrix $\X$ was centered and scaled so
  that the mean of each column was zero and its standard deviation was
  1. The other data matrices were not centered or scaled.

\item We chose $s$, the number of non-zero coefficients. Unless stated
  otherwise we set $s = 20$. We selected the indices $j \in \{1,
  \ldots, p\}$ of the $s$ non-zero coefficients uniformly at random
  among the $p$ variables.

\item We simulated the $s$ non-zero coefficients $b_j$ {\em i.i.d.}
  from some distribution, $h$. We used the following distributions:
  standard normal; uniform on $[-1,1]$; double-exponential (Laplace)
  distribution centered at zero with variance $2\lambda^2$, $\lambda =
  1$ \citep{bda}; $t$-distribution with 1, 2, 4 and 8 degrees of
  freedom; and a point mass (so that all coefficients were the
  same). Unless stated otherwise, we simulated the coefficients from
  the standard normal distribution.

\item Finally, we simulated the responses $y_i = \sum_{j=1}^p x_{ij}
  b_j + e_i$, where $e_i$ was drawn from some noise distribution. We
  used the following noise distributions: normal with mean zero;
  uniform distribution with mean zero; double-exponential (Laplace)
  distribution centered at zero; and $t$-distribution with 1, 2, 4 and
  8 degrees of freedom. In all cases, the variance of the noise
  distribution was adjusted to attain the target PVE; specifically,
  denoting the variance of the noise distribution by $\sigma^2$, we
  set it to $\sigma^2 = \mathrm{Var}(\X\b) \times \frac{1 -
    \mathrm{PVE}}{\mathrm{PVE}}$. Unless stated otherwise, we set
  $\mathrm{PVE} = 0.5$, and simulated $e_i \sim N(0, \sigma^2)$, 
  with $\sigma^2$ adjusted to attain the target PVE of 0.5.

\end{itemize}
We repeated the simulations 20 times for each simulation setting in
Experiments 1--5. The test sets used to evaluate the model fits were
the same size as the training sets and generated using the same
coefficients $\b$.

\subsection{Evaluation}
\label{subsec:evaluation}

Each method returns $\bhatbold$, an estimate of the regression
coefficients. We evaluated this estimate using the (scaled) root
mean squared prediction error in the test data:
\begin{equation}
\mbox{RMSE-scaled}(\y_{\mathrm{test}}, \bhatbold) 
\triangleq
\frac{\textrm{RMSE}(\y_{\mathrm{test}}, \bhatbold)} 
     {\textrm{RMSE}(\bhatbold = {\bm 0})},
\label{eqn:performance_measure}
\end{equation}
where
\begin{equation}
\textrm{RMSE}(\y_{\mathrm{test}}, \bhatbold) 
\textstyle \triangleq \frac{1}{\sqrt{n}}
\norm{\y_{\mathrm{test}} - \X_\mathrm{test}\bhatbold}.
\label{eqn:rmse}
\end{equation}
% (Here, we are assuming that $\X_{\mathrm{test}}$ and
% $\y_{\mathrm{test}}$ were centered in the same way as $\X$ and $\y$ to
% account for a non-zero intercept.) A problem with reporting the ``raw''
% RMSE metric directly is that the RMSE will vary greatly depending on
% the difficulty of the data set. This can make it difficult to
% summarize results across simulations. For example, data sets simulated
% from dense models with low PVE will are the most challenging, and
% therefore the RMSE for these data sets will typically be higher than
% data sets simulated from sparse models with high PVE. Therefore, to
% provide a metric that is more comparable across different simulation
% setting, we instead used the following scaled RMSE measure,
$\textrm{RMSE}(\hat{\b} = {\bm 0}) \triangleq \sigma/\sqrt{1 -
  \textrm{PVE}}$ denotes the expected RMSE for the ``null predictor'',
$\bhatbold = {\bm 0}$. The value of ``RMSE-scaled'' will vary from
$\sqrt{1 - \mathrm{PVE}}$ (for the ``oracle'' predictor) to
approximately 1 (for the null predictor); however, values greater than
1 are possible if $\bhatbold$ performs worse than the null predictor.
% However, for understanding the results, it is important to keep in
% mind that an increase (or decrease) in the scaled RMSE does not
% necessarily imply an increase (or decrease) in the (absolute) RMSE.

\subsection{Results}
\label{subsec:predperformance}

We present the results of Experiments 1--5 in Sections
\ref{subsec:varying-sparsity}--\ref{subsec:varying-noise-dist}, and we
give a high-level summary in
Section~\ref{subsec:relative_performance}.

\subsubsection{Experiment 1---Varying the Sparsity Level}
\label{subsec:varying-sparsity}

\begin{figure}[t]
\centering
\includegraphics[width=0.975\textwidth]{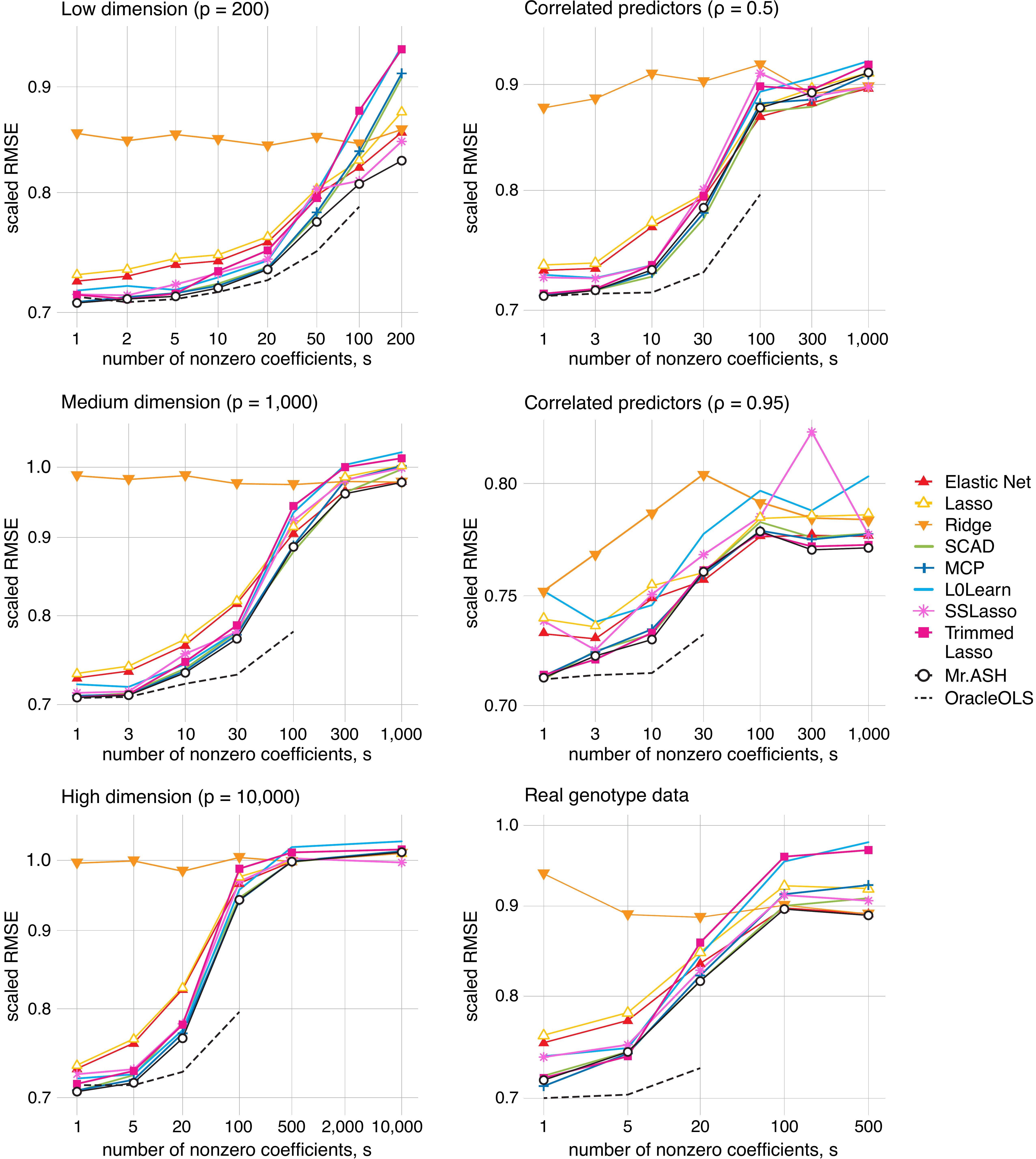}
\caption{Results from Experiment 1 in which the sparsity
    level, $s$, was varied. Each point shows the prediction error
    (scaled RMSE) averaged over 20 simulations.}
% The ``Oracle OLS'' result is only shown in the plots
%    when it provides a lower bound on the prediction error. (For some
%    larger settings of $s$, the Oracle OLS performed poorly.)}
\label{fig:plrsparsity}
\end{figure}

\begin{figure}[t]
\centering
\includegraphics[width=\textwidth]{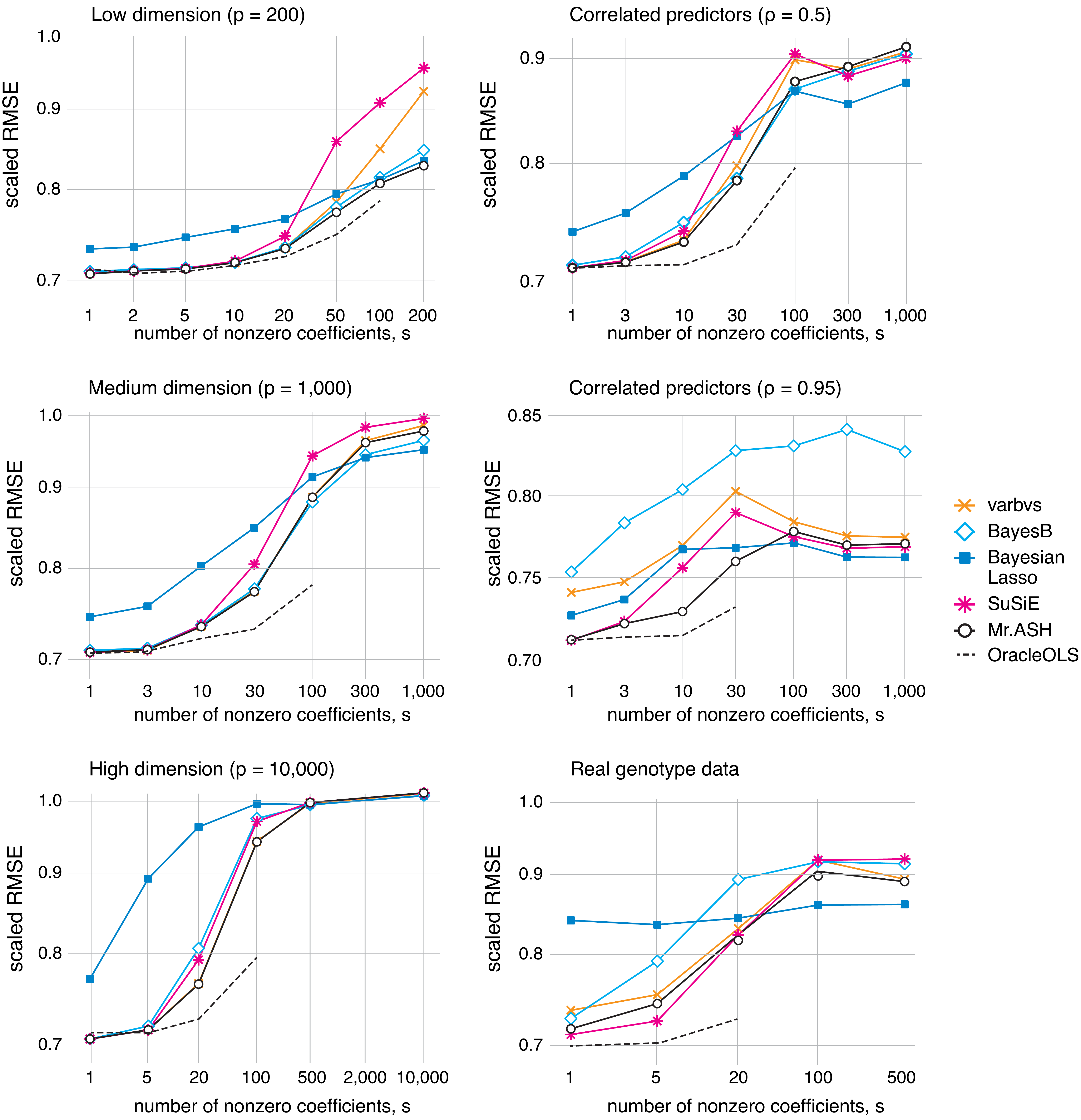}
\caption{Results from Experiment 1 in which the sparsity
    level, $s$, was varied. Each point shows prediction error
    (scaled RMSE) averaged over 20 simulations. 
% The ``Oracle OLS'' result is only shown in the plots when it
% provides a lower bound on the prediction error. (For some larger
% settings of $s$, the Oracle OLS performed poorly.
The \mrash{} results are the same as in
Figure~\protect\ref{fig:plrsparsity} and provide a common point of
reference.}
\label{fig:blrsparsity}
\end{figure}

In the first set of simulations, we varied $s$, the number of non-zero
predictors, which controls the model sparsity when $p$ is fixed.  The
results of these simulations, summarized in Figures
\ref{fig:plrsparsity} and \ref{fig:blrsparsity}, highlight differences
in performance between the methods that were better suited to a
particular level of sparsity versus the methods that better adapted to
different levels of sparsity. For example, \susie, \lzero{} with the
$L_0$ penalty and the Trimmed Lasso better adapted to sparse settings
and often performed poorly in dense settings; by contrast, ridge
regression and the Bayesian Lasso are better adapted to dense
settings, and as expected performed worse in sparse settings.

Other methods performed more consistently across different sparsity
levels. In particular, the \lasso{} and Elastic Net performed somewhat
similarly, with the Elastic Net usually performing slightly better,
but in many settings there was a noticeable gap in performance between
these two methods compared with the best performing method. The
nonconvex, penalty-based methods \mcp{} and \scad{} performed similarly to
one another, and were competitive in many settings, except in
denser-signal data sets in some scenarios (e.g., in the low-dimension
settings, with $p = 200$ non-zero coefficients). The Spike-and-Slab
Lasso (SSLasso) with the adaptive penalty performed competitively
across a range of sparse to dense settings when predictor variables
were independent, but it performed less well (and sometimes very
poorly) in settings with correlated predictors.

\mrash{} was competitive at all sparsity levels, consistently
achieving the best performance or close to the best in almost all
simulation settings. \mrash{} tended to outperform other methods in
data sets with correlated predictors, with a couple of exceptions:
the Bayesian Lasso had better prediction accuracy in settings with
the densest signals, and \susie{} was better in some of the sparse
simulation settings with genotype data. One possible explanation for
this is that the mean-field variational approximation used by
\mrash{} is less appropriate in settings with strongly correlated
predictors, whereas the Bayesian Lasso and \susie{} do not have this
issue (Bayesian Lasso uses MCMC to compute the posterior estimate;
and \susie{} uses a different variational approximation that can deal
with strong correlations).
  
Overall, these results illustrate the versatility of \mrash{} and the
effectiveness of the VEB approach to adapt to different sparsity
levels by adapting the prior---and therefore penalty---to the data.

Two other results stand out. First, \bayesb{} performed poorly in the
simulations with correlated predictors. In principle, correlated
predictors should not cause problems for Bayesian methods such as
\bayesb.  We suspect that this reflects failure of the MCMC to
converge and that the performance of BayesB (and the Bayesian Lasso)
would improve with longer MCMC runs.

Second, although \varbvs{} is based on the same variational
approximation as \mrash, its prediction accuracy was generally worse
than \mrash. This is particularly evident in some settings with dense
signals, probably because the default settings in \varbvs{} are
designed to favor sparse priors.  Also, \varbvs{} performed worse than
\mrash{} in some data sets with correlated predictors, probably
because \varbvs{} does not put as much effort into
initialization. (See Appendix~\ref{subsec:initorder} for
investigations of the impact of initialization on the performance of
\mrash{}.)

\subsubsection{Experiment 2---Varying Total Signal Strength}
\label{subsec:varying-total-signal}

\begin{figure}[t]
\centering
\includegraphics[width=\textwidth]{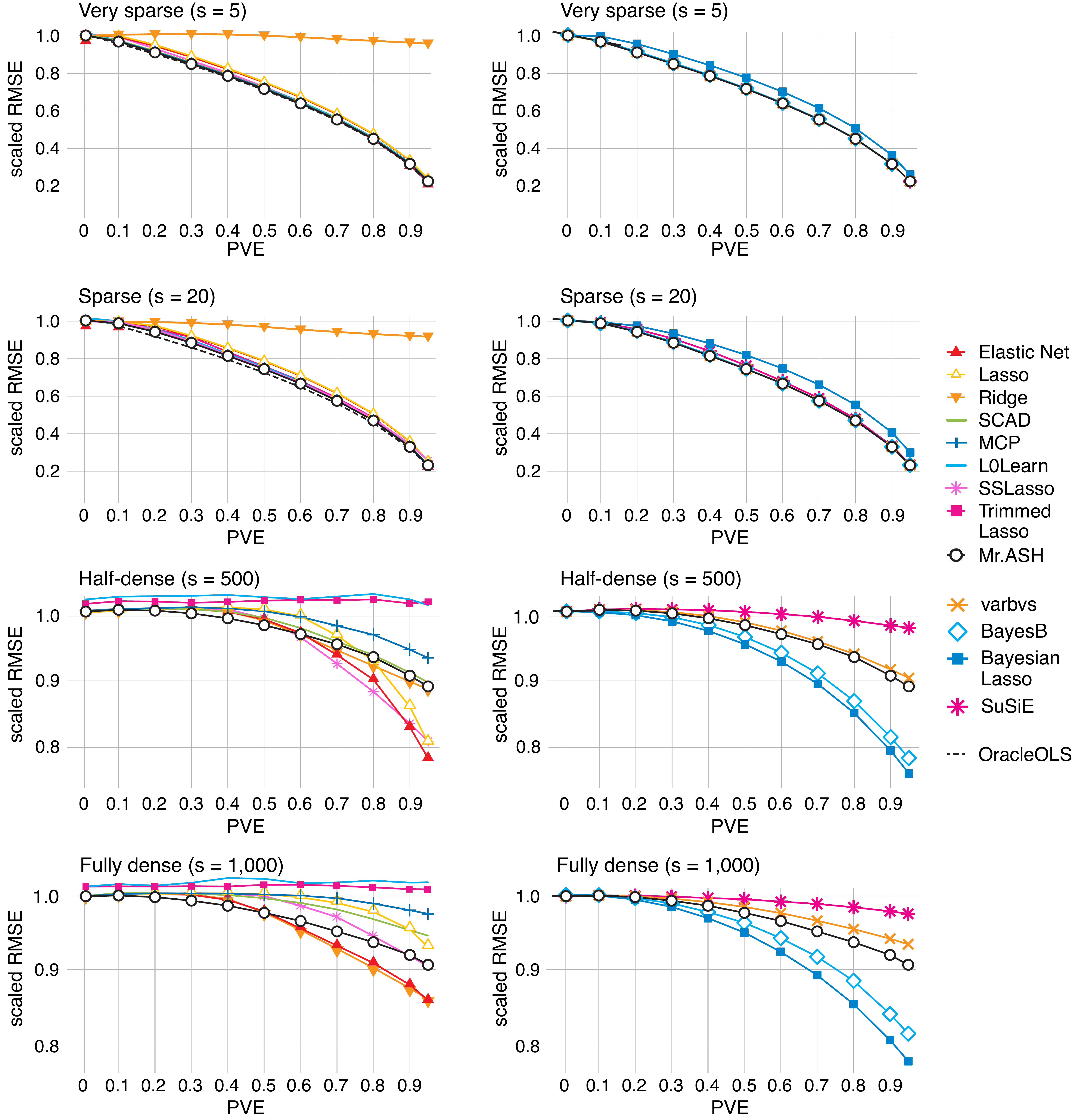}
\caption{Results from Experiment 2 in which total signal
    strength (PVE) was varied. Each point shows prediction error
    (scaled RMSE) averaged over 20 simulations. Left panels show PLR
    methods; right panels show Bayes-based methods. The \mrash{}
    results are included in all plots to provide a common reference
    point.}
\label{fig:plr-npveshape-1}
\end{figure}

We did not expect strong systematic differences in performance as the
total signal strength (PVE) varied. Results for sparse simulations
(top half of Figure~\ref{fig:plr-npveshape-1}) generally matched these
expectations; the performance curves of the different methods
generally did not cross as the PVE varies. And, among the methods
compared, \mrash{} consistently performed among the best. However,
results in dense simulations (bottom half of
Figure~\ref{fig:plr-npveshape-1}), showed a different pattern: some
methods that performed competitively at moderate PVE were no longer
competitive at high PVE. This included \mrash, which was unexpected
because \mrash{} should, in principle, adapt well to both sparse and
dense signals. Further, while \mrash{} includes ridge regression as a
special case, the ridge regression estimates yielded better
predictions than \mrash{} in these dense data sets, suggesting a
failure of the variational EB approach to appropriately adapt the
prior.  We believe that this failure occurred because the fully
factorized (mean-field) variational method has a tendency in some
settings to favor sparse priors over dense priors; under sparse
priors, the true posterior distribution comes close to factorizing so
the gap between the ELBO and the true evidence is small ({\em i.e.},
the ELBO is a tight lower bound), whereas under dense priors with
large $p$ there are stronger dependencies in the posterior, especially
when the PVE is large. Thus, the factorization assumption is more
strongly violated and the ELBO underestimates the evidence more. Since
our EB approach seeks to optimize the ELBO in place of the true
evidence, this creates a bias towards estimating sparse priors rather
than dense priors. In this sense, the VEB approach could be
characterized as leaning towards the ``bet on sparsity'' principle
\cite[][p.~610]{hastie2009elements}, which argues that one should
``use a procedure that does well in sparse problems, since no
procedure does well in dense problems.'' These simulations therefore
suggest a situation (dense signal, high PVE, $p = 2n$) where this
principle fails.
  
The good performance of \bayesb{} in these simulations with dense
signals and high PVE suggests an advantage of MCMC over the
variational approximation in dense scenarios in which many predictors
have a small effect on the regression outcome.  Consistent with
Experiment 1, ridge regression produced comparatively poor predictions
in sparse settings, whereas L0Learn, SuSiE and the Trimmed Lasso had
worse accuracy in dense settings.

\subsubsection{Experiment 3---Varying Signal Distribution} 
\label{subsec:varying-signal-dist}

\begin{figure}[t]
\centering
\includegraphics[width=\textwidth]{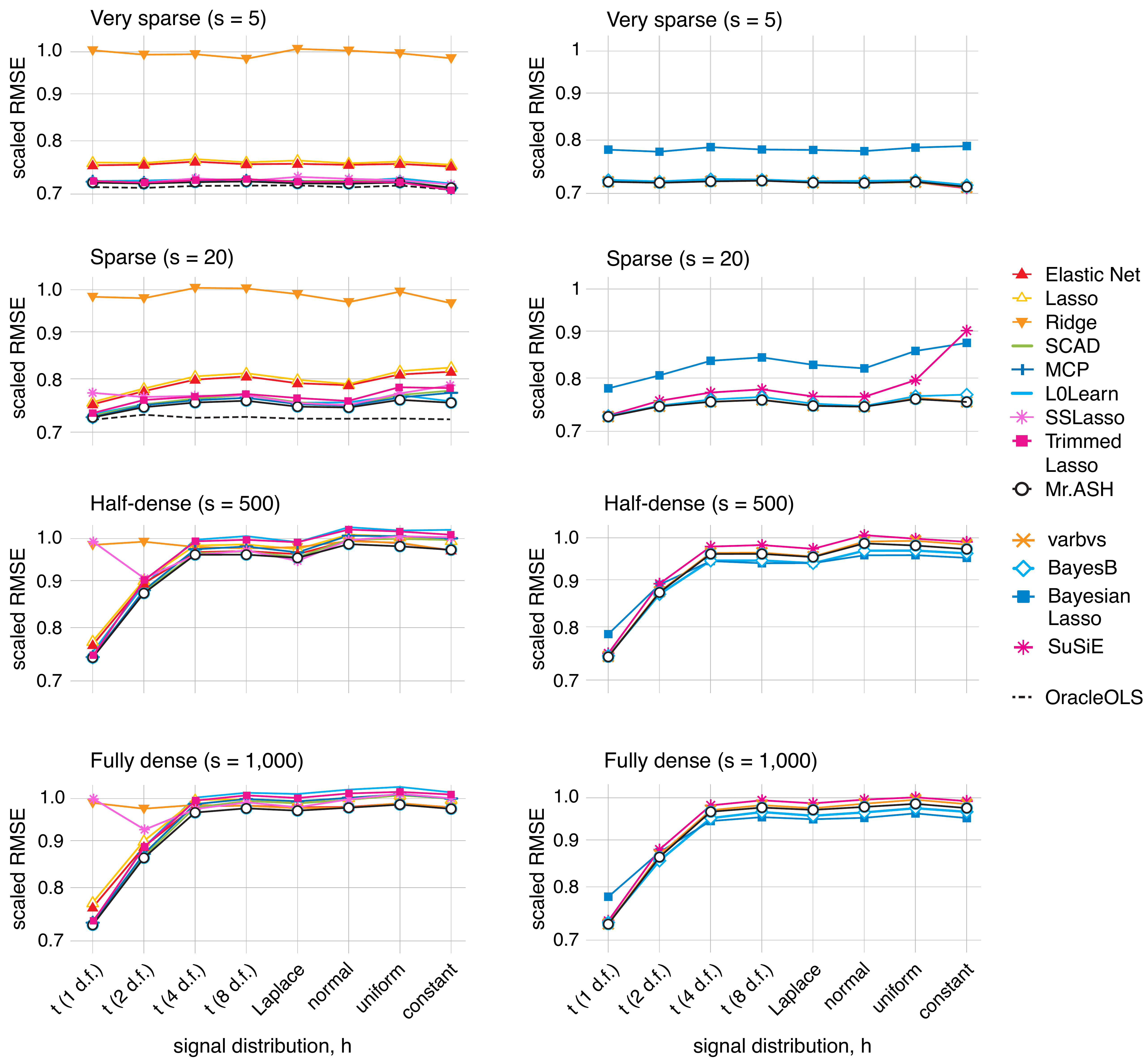}
\caption{Results from Experiment 3 in which the signal
    distribution ($h$) was varied. Each point shows the prediction
    error (scaled RMSE) averaged over the 20 simulations at that
    setting. Left panels show PLR methods; right panels show
    Bayes-based methods. The \mrash{} results are included in all the
    plots to provide a common reference point.}
% The results are arranged along the horizontal axis in each plot in
% such a way that the simulations with the greatest variation in the
% coefficients are on the left.
\label{fig:plr-npveshape-2}
\end{figure}

For most methods, the distribution used to simulate the
coefficients, $h$, had only a small impact on performance
(Figure~\ref{fig:plr-npveshape-2}). There were a few exceptions to
this. For example, in dense settings the ridge regression method
struggled with long-tailed effect-size distributions such as the
$t$-distribution with 1 degree of freedom. Consider that,
when $h$ is long-tailed, often a small number of coefficients will
dominate, so the normal prior in ridge regression is poorly suited
for this case. The SSLasso also performed poorly in this setting
for reasons that remain unclear to us.

\subsubsection{Experiment 4---Varying the Number of Predictors}
\label{subsec:varying-p}

\begin{figure}
\centering
\includegraphics[width=0.95\textwidth]{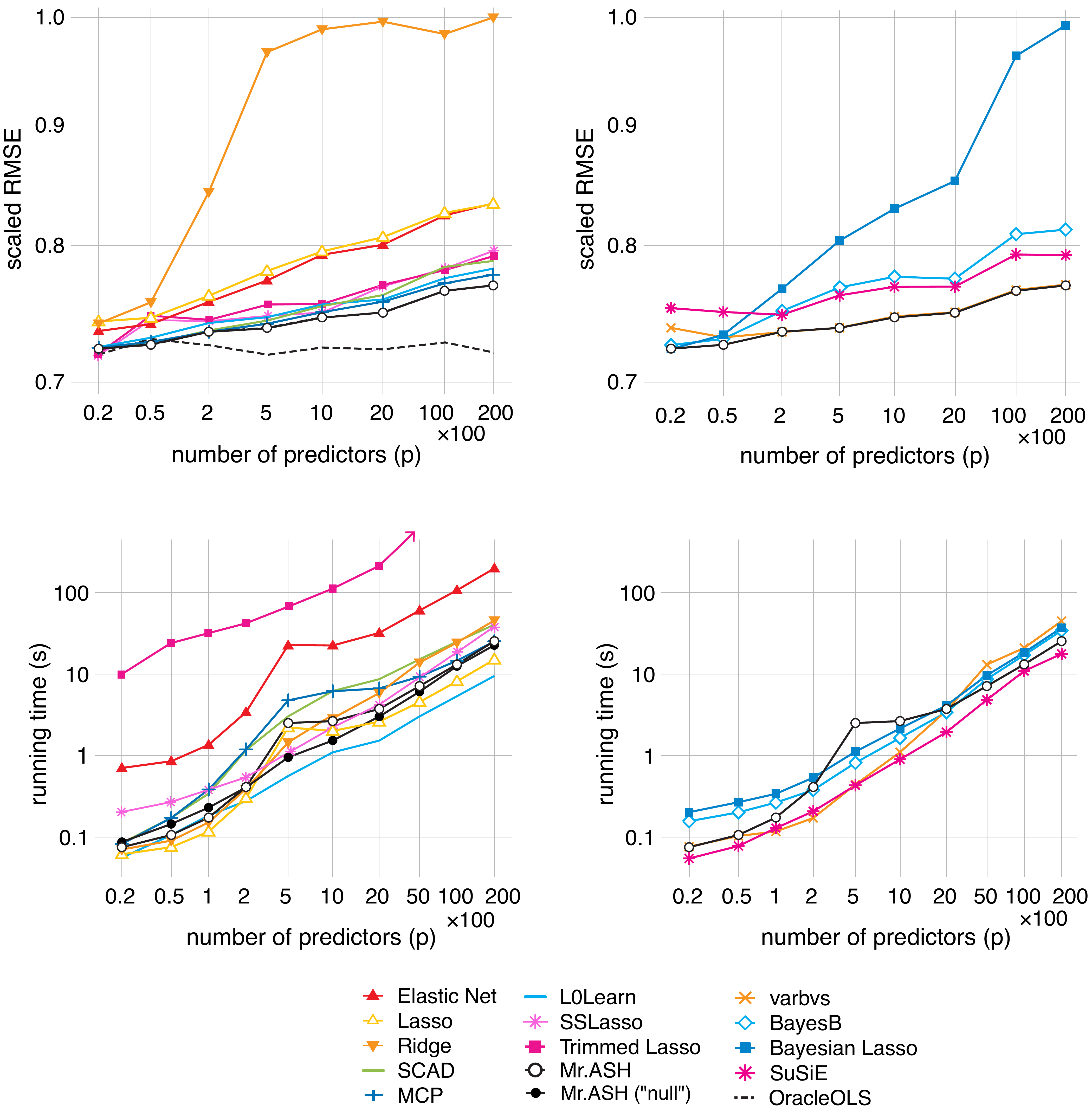}
\caption{Results from Experiment 4 in which the number of
    predictors, $p$, was varied. Each point shows the scaled RMSE (top
    row) and the running time (bottom row) averaged over 20
    simulations. Left panels show PLR methods; right panels show
    Bayes-based methods.  The \mrash{} results are included in all
    plots to provide a common reference point. Running times for all
    methods except Trimmed Lasso are from running the methods in R 3.5.1
    \citep{R} on a machine with a quad-core 2.6 GHZ Intel Core i7
    processor and 16 GB of memory. R was installed from macOS
    binaries and we used the BLAS libraries that were distributed with
    R. The Trimmed Lasso was run using MATLAB 9.13 on machines with 4
    Intel Xeon E5-2680v4 (``Broadwell'') processors and 24 GB of
    memory.  The \mrash{} running times include running the \lasso{} to
    initialize the estimates; the running times for \mrash{} with a
    ``null'' initialization ($\b = {\bm 0}$) are also shown for
    comparison.}
\label{fig:highdim}
\end{figure}

This experiment assessed how prediction accuracy and computational
effort change with the number of predictors, $p$
(Figure~\ref{fig:highdim}). In general, running times for different
methods increased similarly with $p$. For large $p$, SuSiE, L0Learn
and the Lasso were the fastest methods. The Elastic Net was
considerably slower than other methods because we tuned both of the
Elastic Net parameters by CV whereas other methods tuned by CV
involved tuning only one parameter. (The Elastic Net could be run
faster by tuning only one parameter but at the risk of losing accuracy
in some settings.) Fitting the \mrash{} prior involved tuning a large
number of parameters, but because it tuned these parameters via VEB
rather than by CV, \mrash{} ended up being roughly as fast as methods
that tune a single parameter by CV. This is an important benefit of
the VEB approach.

The long running times for the Trimmed Lasso were due to running the
algorithm at several different settings of its ``target sparsity
level'' parameter (`$k$'), and the Trimmed Lasso was slow when $k$
was large (even with the settings that were suggested in the software
documentation for larger data sets). When $k$ was small, say, less
than 10, the Trimmed Lasso running times were comparable to the other
methods.

An important detail is that a large fraction of the effort in running
\mrash{} was due to running the \lasso{} (which was used to initialize
\mrash{}). The \lasso{} initialization often greatly reduced the
number of iterations required for the \mrash{} coordinate ascent
updates to converge, so the total running time of \mrash{} with Lasso
initialization was often not much greater than \mrash{} with null
initialization. 

Although MCMC methods have a reputation for being slow, here the
MCMC-based methods---the Bayesian Lasso and \bayesb---had similar
running times to other methods. The running time of these methods also
increased approximately linearly with $p$ because the defaults in the
software implementations set the number of iterations proportional to
$p$. (The per-iteration cost is independent of $p$ when the model
configurations are sparse.)  However, as $p$ increased, the relative
prediction performance of these methods got worse. For the Bayesian
Lasso, this was probably due to the fact that we fixed $s$ (the number
of non-zero coefficients) to simulate the data sets, so simulations
with larger $p$ were based on sparser models, and the Bayesian Lasso
tends to be less competitive in sparser settings. For \bayesb, the
reduction in prediction accuracy may instead reflect a failure of the
Markov chain to adequately explore the posterior distribution, and
perhaps better performance could have been achieved by running the
MCMC longer (at increased computational cost). Nonetheless, it is
interesting that \bayesb{} obtained better prediction accuracy than
the \lasso{} at roughly the same computational effort.

To summarize, \mrash{} consistently achieved the best
prediction accuracy in these simulations, with computational effort
comparable with the fastest methods such as L0Learn
and Lasso.

\subsubsection{Experiment 5---Varying Noise Distribution}
\label{subsec:varying-noise-dist}

\begin{figure}[!t]
\centering
\includegraphics[width=\textwidth]{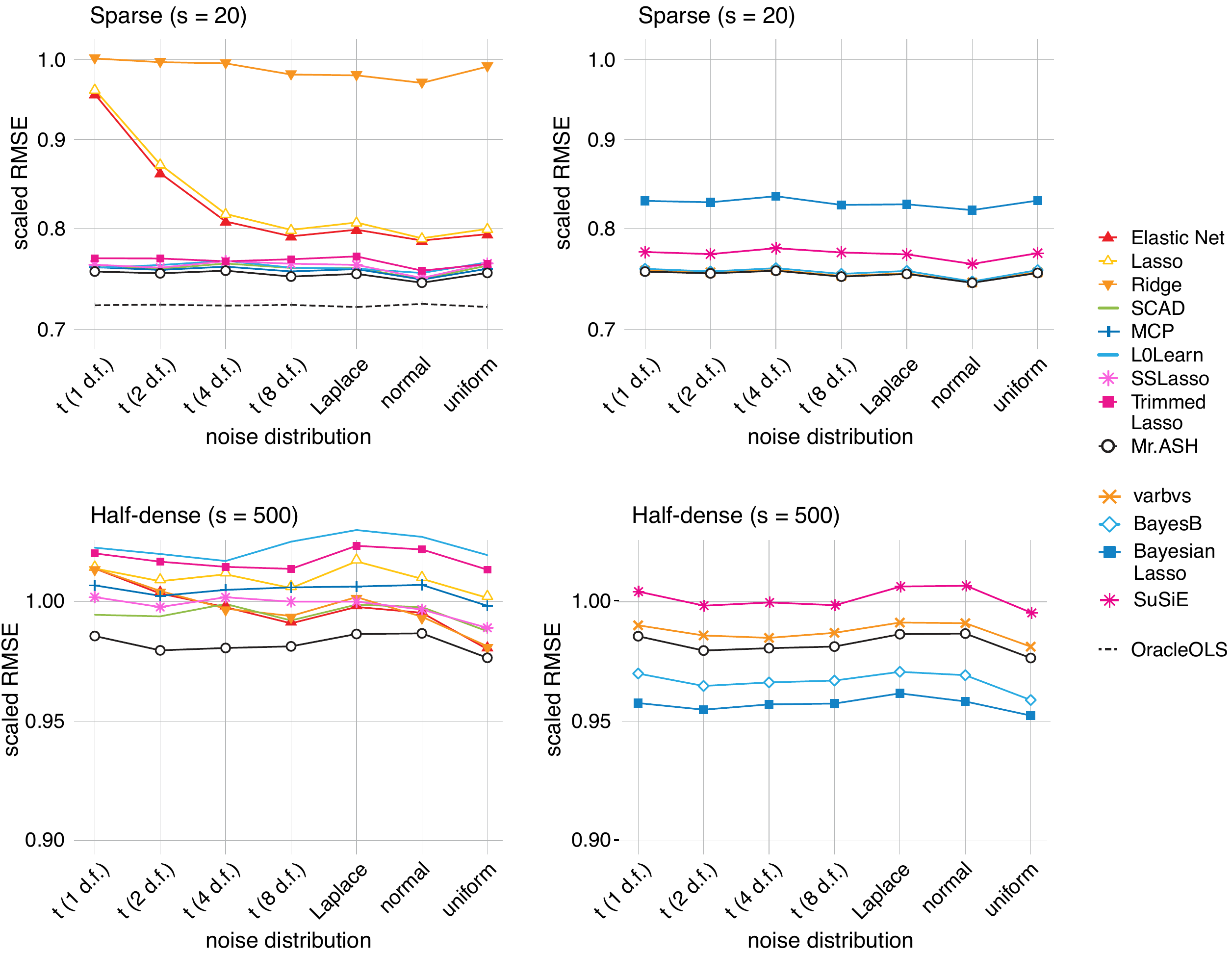}
\caption{Results from Experiment 5 in which the noise
    distribution was varied. Each point shows the prediction error
    (scaled RMSE) averaged over 20 simulations. Left panels show PLR
    methods; right panels show Bayes-based methods. The \mrash{}
    results are included in all plots to provide a common reference
    point.}
\label{fig:noisedist}
\end{figure}

In our final set of experiments, we simulated data sets with
different noise distributions. Reassuringly, most methods were
largely insensitive to the noise distribution
(Figure~\ref{fig:noisedist}). However, the Lasso and Elastic Net both
performed poorly in sparse settings when the noise was very
heavy-tailed ($t$ distribution with small degrees of freedom). We do
not have an explanation for this result, which we have not seen 
previously.

\subsubsection{Summary of the Results}
\label{subsec:relative_performance}

\begin{figure}[t]
\centering
\includegraphics[width=0.84\textwidth]{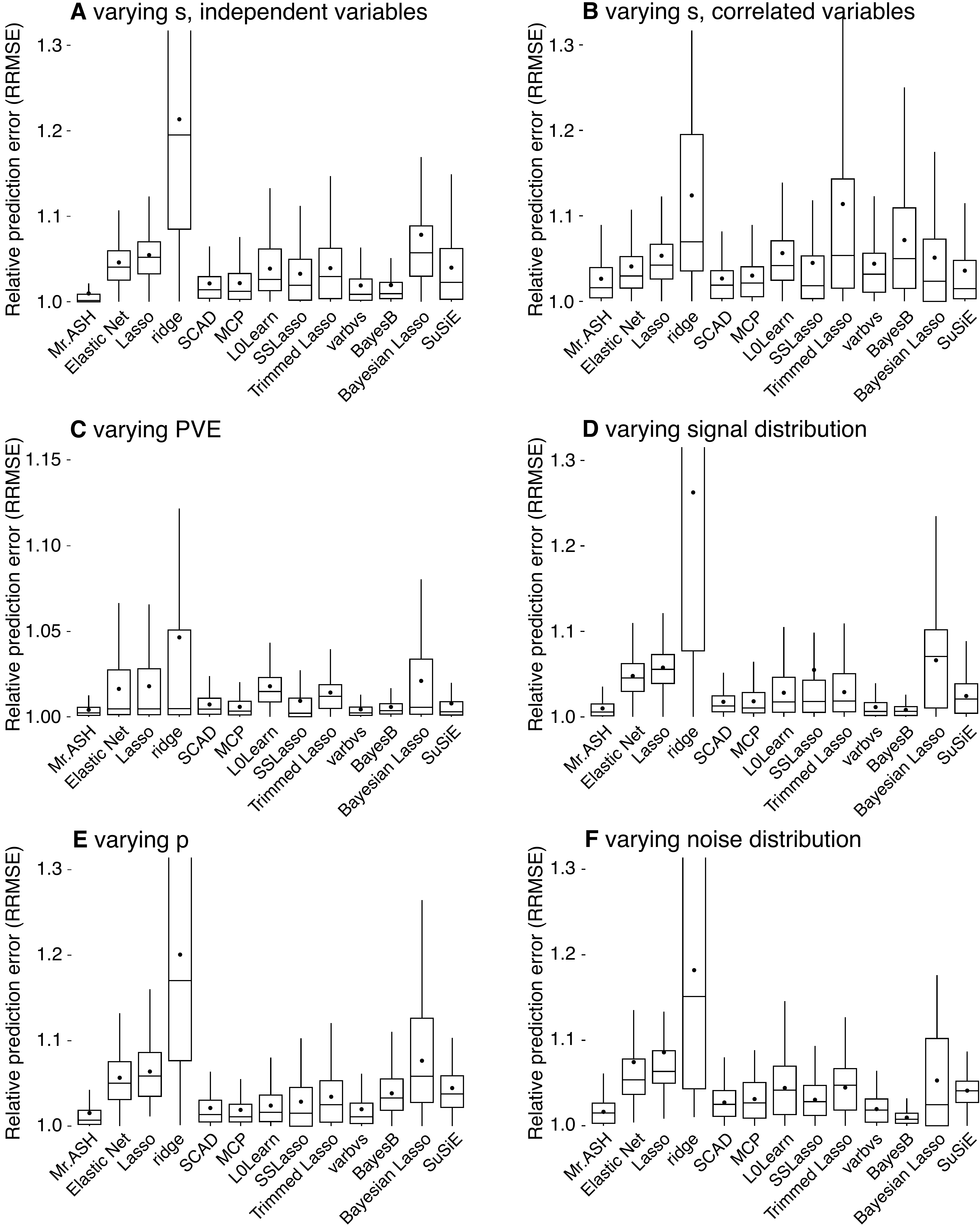}
\caption{Summary of results from Experiments 1--5. The
    simulation results are summarized slightly differently in this
    figure to emphasize common trends. The boxplots show the
    distribution of RMSEs relative to the best performing method in
    each simulation (see the text for details). The horizontal line
    inside each box depicts the median; the dot depicts the mean; and
    the upper and lower lines depict the interquartile range.}
\label{fig:rrmse}
\end{figure}

The five experiments highlighted some differences in performance and
behaviour among the multiple regression models. In
Figure~\ref{fig:rrmse}, we give a higher level summary of the results
across all five of the experiments. To produce this summary, for each
simulation $t$ we calculated the relative prediction accuracy as the
ratio of the RMSE to the best RMSE achieved in that simulation among
all the methods compared, $\mathrm{RRMSE}_{tm} \triangleq
\mathrm{RMSE}_{tm}/ \min_{m'} \mathrm{RMSE}_{tm'}$, where
$\mathrm{RMSE}_{tm}$ is the root mean squared error \eqref{eqn:rmse}
generated by model $m$ for the test set in simulation $t$. Defined in
this way, the RRMSE can never be smaller than 1 and the most accurate
method in a given simulation has an RRMSE of 1.

\begin{table}[t]
\centering
\begin{tabular}{l@{\;\;}r}
method & running time (s) \\
\midrule
L0Learn & 8.12 \\
\mrash (``null'') & 14.96 \\
SSLasso & 16.19 \\
Lasso & 17.98 \\
Bayesian Lasso & 18.24 \\
\mrash & 18.56 \\
MCP & 21.19 \\
SuSiE & 22.54 \\
BayesB & 23.77 \\
ridge regression & 32.88 \\
SCAD & 33.59 \\
varbvs & 52.26 \\
Elastic Net & 223.66 \\
Trimmed Lasso & 609.54 
\end{tabular}
\caption{Average running times in Experiments 1--3.}
% The \mrash{} running time includes fitting the \lasso{} model to
% initialize (average running time for \mrash{} with null
% initialization was 14.96 s). Methods are sorted by average running
% time.
\label{table:comptime}
\end{table}

Figure~\ref{fig:rrmse} highlights the consistently good
prediction accuracy of \mrash{} compared with the other methods across
a range of settings; in all the panels (A--F), \mrash's performance was
the best, or close to the best, among all methods compared, and was
rarely much worse than the best method. The benefits of \mrash{}
were more mixed in the simulations with correlated predictors (Panel
B). Yet, even in these simulations, \mrash{} remained competitive,
achieving an average prediction accuracy that was among the
best. This good accuracy was also obtained efficiently with a
computational effort that was not much higher than the fastest
methods such as \lzero{} and the \lasso{}
(Table~\ref{table:comptime}).

\section{Discussion}
\label{sec:conclusion}

We have presented a new VEB method for multiple linear regression with
a focus on fast and accurate prediction. This VEB method combines
flexible shrinkage priors with variational methods for efficient
posterior computations. Variational methods and EB methods are
sometimes criticized because of their tendency to understate
uncertainty compared with ``fully Bayesian'' methods; see
\cite{morris1983parametric}, \cite{wang2005inadequacy} and references
therein for discussion. However, for some applications uncertainty is
of secondary importance compared with speed and accuracy of point
estimates. For example, speed and accuracy is often important when
multiple regression is used simply to build an accurate predictor for
downstream use (see \citealt{gamazon2015gene} for one such
application). Our VEB approach seems particularly attractive for such
uses.

A natural next step would be to produce similar VEB methods for
non-Gaussian ({\em i.e.}, generalized) linear models
\citep{mccullagh1989generalized}. Extension of our methods to logistic
regression should be possible via additional approximations that allow
for efficient analytic computations \citep{carbonetto2012scalable,
  varbvs, jaakkola2000bayesian, marlin-2011,
  pattern-recognition-machine-learning,
  wang2013variational}. Extension to other types of outcome
distributions and link functions should also be possible but may
require more work.

Our work also illustrates the benefits of an EB approach in an
important and well studied statistical problem. While there is much
theoretical \citep{johnstone2004needles} and empirical
\citep{efron2008microarrays} support for the benefits of EB
approaches, EB approaches have not been widely adopted outside of
specific research topics such as wavelet shrinkage
\citep{johnstone2005empirical} and moderated estimation in gene
expression studies \citep{smyth2004linear, lu2016variance, zhu-2019}.
Recent work has highlighted the potential for EB methods in other
applications, including smoothing non-Gaussian data
\citep{xing2016flexible}, multiple testing \citep{stephens2016false,
  sun2018solving, urbut2019flexible, gerard2020empirical}, matrix
factorization \citep{wang2018empirical} and additive models
\citep{wang2018simple}. We hope that these examples, including our
work here, will inspire readers to apply EB approaches to new
problems.

\acks{We thank Gao Wang for help with processing the the GTEx data,
  and for helpful discussions. This work was supported by NIH grant
  R01HG002585 to Matthew Stephens. We also thank the anonymous
  reviewers for their constructive suggestions for improvement and the
  staff at the Research Computing Center for providing the
  high-performance computing resources used to implement the numerical
  experiments.}

\newpage

\appendix

\def\section{\@startsiction{section}{1}{\z@}{-0.24in}{0.10in}
             {\large\bf\raggedright Appendix }}

\section{Shrinkage Operators of Commonly Used Penalties} 
\label{appendix:penalties}

\begin{table}[h!]
\centering
\begin{tabular}{@{}l@{\;\;}l@{\;\;}l@{}}
\hspace*{1ex}method &
penalty function $\rho(t)$ &
shrinkage operator $S_{\rho}(t)$ \\ \midrule
$\begin{array}{l}
\mbox{normal shrinkage (ridge} \\
\mbox{regression, $L_2$ penalty)}
\end{array}$
& $\lambda t^2/2$
& $\frac{t^2}{1 + \lambda}$ \\[1em]
$\begin{array}{l}
\mbox{hard thresholding (best} \\
\mbox{subset, $L_0$ penalty)}
\end{array}$
& $\lambda \times \mathbb{I}\{|t| > 0\}$
& $\left\{\begin{array}{ll}
t & \textrm{if } t < -\lambda, \\
t & \textrm{if } t > \lambda, \\
0 & \mbox{otherwise}
\end{array}\right.$ \\[2.5em]
$\begin{array}{l}
\mbox{soft thresholding} \\
\mbox{(Lasso, $L_1$ penalty)}
\end{array}$ &
$\lambda|t|$ &
$S_{\mathrm{soft}, \lambda} \triangleq
\left\{\begin{array}{ll}
t + \lambda & \mbox{if } t < -\lambda, \\
t - \lambda & \mbox{if } t > \lambda, \\
0 & \mbox{otherwise} 
\end{array}\right.$ \\[2.5em]
\hspace*{1ex}Elastic Net
& $(1 - \eta)\lambda t^2/2 + \eta\lambda|t|$
& $\begin{array}{l}
S_{\mathrm{soft},\eta\lambda/a}(t/a), \\
a = 1 + (1 - \eta)\lambda
\end{array}$ \\[2em]
$\begin{array}{l}
\mbox{Minimax Concave} \\
\mbox{Penalty}
\end{array}$
& $\left\{\begin{array}{ll}
\lambda |t| - t^2/(2\eta) & \textrm{if } |t| \leq \eta \lambda, \\
\eta\lambda^2/2 & \mbox{otherwise}
\end{array}\right.$
& $\left\{\begin{array}{ll}
\frac{S_{{\rm soft}, \lambda}(t)}{1 - 1 / (\eta - 1)} &
\textrm{if } |t| \leq \eta \lambda, \\
t & \mbox{otherwise}
\end{array}\right.$ \\[2.5em]
$\begin{array}{l}
\mbox{Smoothly Clipped} \\
\mbox{Absolute Deviation}
\end{array}$
& $\left\{\begin{array}{ll}
\lambda|t| &\textrm{if } |t| \leq 2 \lambda, \\
\lambda^2(\eta + 1)/2 & \textrm{if } |t| > \eta \lambda, \\
\frac{\eta\lambda|t| - (t^2 + \lambda^2)/2}{\eta - 1}
& \textrm{otherwise}
\end{array}\right.$
& $\left\{\begin{array}{ll}
S_{\mathrm{soft}, \lambda}(t) & \textrm{if } |t| \leq 2 \lambda, \\
t &\textrm{if } |t| > \eta \lambda, \\
\frac{S_{{\rm soft}, \eta\lambda/(\eta - 1)}(t)}{1 - 1/(\eta - 1)}
& \textrm{otherwise}
\end{array}\right.$
\end{tabular}
\caption{Some commonly used penalty functions and their corresponding
  shrinkage operators.}
\label{table:penalty}
\end{table}

\section{Additional Notes on the Methods Compared}
\label{appendix:methods_compared_details}

Here we give additional details about how the methods were applied to
the simulated data sets.

\paragraph{Ridge regression.} 

We used function {\tt cv.glmnet} from the {\tt glmnet} R
package to choose the penalty strength by CV. Specifically, we
called {\tt cv.glmnet} with {\tt alpha = 0}, {\tt intercept = TRUE}
and {\tt standardize = FALSE}; all other settings were kept at their
defaults. The setting of the penalty strength parameter $\lambda$
minimizing the mean CV error ({\tt lambda.min}) was
used to make predictions.

\paragraph{Lasso.}

We fit Lasso models in the the same way that we fit ridge
regression models using {\tt glmnet}, except that we set the Elastic
Net mixing parameter to {\tt alpha = 1}.

\paragraph{Elastic Net.}

Elastic Net models were fit similarly to ridge regression and Lasso
models, again using the {\tt cv.glmnet} interface from the {\tt
  glmnet} package. The difference is that the Elastic Net involves two
tuning parameters: the penalty strength parameter $\lambda$ and the
``mixing'' parameter $\alpha$. {\tt glmnet} does not provide an
automated way to chose $\alpha$ by CV, so we ran {\tt cv.glmnet} for
11 settings of $\alpha$ ranging from 0 to 1 then we chose $(\alpha,
\lambda)$ minimizing the 10-fold CV error.

\paragraph{SCAD and MCP.}

We called the {\tt cv.ncvreg} function from the R package {\tt ncvreg}
which performs $k$-fold CV to select the regularization parameter. We
called {\tt cv.ncvreg} with {\tt nfolds = 10} and {\tt penalty =
  "SCAD"} or {\tt penalty = "MCP"}. We kept other settings at
their defaults. By default, {\tt cv.ncvreg} standardized ${\bf X}$ and
added an intercept to the model.

\paragraph{L0Learn.}

We called the {\tt L0Learn.cvfit} function from the R package {\tt
  L0Learn} which performs $k$-fold CV to select the penalty strength
parameter $\lambda_0$. We called {\tt L0Learn.cvfit} with {\tt penalty
  = "L0"} and {\tt nFolds = 10}. We chose the setting of $\lambda$
with the smallest (mean) CV error.  We kept other settings at their
defaults. By default, {\tt L0Learn.cvfit} included an intercept in the
model.

\paragraph{SSLasso.}

We called the {\tt SSLASSO} function from the R package {\tt SSLASSO}
which fits coefficients paths for spike-and-slab linear regression
models over a grid of values for the $\lambda_0$ regularization
parameter. We used the ``adaptive'' variant of the Spike-and-Slab
Lasso which was recommended over the ``separable'' variant. This
function automatically standardizes ${\bf X}$ and the model includes
an intercept. {\tt SSLASSO} did not provide an automated way to select
$\lambda_0$ so we performed 5-fold CV and chose the setting of
$\lambda_0$ minimizing the CV error. We did not perform CV to choose
the ``spike'' penalty parameter $\lambda_1$; \cite{rovckova2018spike}
showed the performance is less sensitive to the choice of $\lambda_1$.

\paragraph{Trimmed Lasso.}

We downloaded the MATLAB code for the Trimmed Lasso from
\url{https://github.com/tal-amir/sparse-approximation-gsm} and we
compiled the MEX file using gcc 10.2.0. We called function {\tt
sparse\_approx\_gsm\_v1\_22} with the following settings: {\tt
sparse\_approx\_gsm\_v1\_22(X,y,k,`profile',`fast')} in which the
target sparsity level {\tt k} was one of $\{1, 5, 20, 100, 500,
2000, 10000 \}$. The sparsity level was chosen by 5-fold
CV, taking the setting that minimized the average
CV error.

\paragraph{BayesB and Bayesian Lasso}

We called the {\tt BGLR} function from the {\tt BGLR} package, which
simulates the posterior using a Gibbs sampler. We called {\bf BGLR}
with the following settings: {\tt standardize = FALSE}, {\tt nIter =
  1500}, {\tt burnIn = 500} and {\tt model = "BayesB"} or {\tt model =
  "ML"}. We kept other settings at their defaults. By default, an
intercept was included in the model.

\paragraph{varbvs.}

We called the {\tt varbvs} function from the {\tt varbvs}
package which fits an approximate posterior distribution for a
Bayesian variable selection model using variational inference
methods. All settings were kept at their defaults. The model
included an intercept and ${\bf X}$ was not standardized.

\paragraph{SuSiE.}

We called the {\tt susie} function from the {\tt susie}
package which fits a SuSiE model using the iterative Bayesian
stepwise selection (IBSS) algorithm. We called {\tt susie} with {\tt
standardize = FALSE} and we set the upper bound on the number of
``single effects'' to 20.

\section{Additional Experiments}

\subsection{Simulations with $p < n$} 
\label{subsec:p<n}

\begin{figure}[t]
\centering
\includegraphics[width=0.4\textwidth]{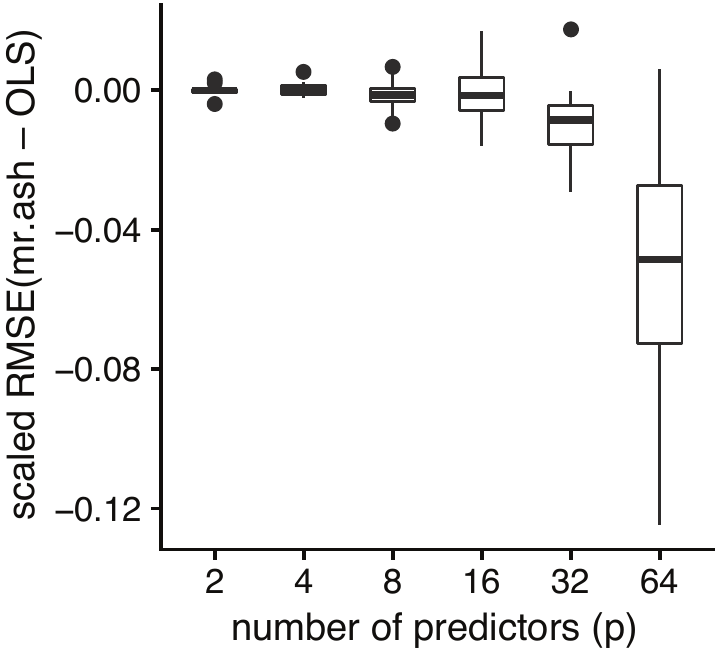}
\caption{Comparison of \mrash{} vs. OLS in data sets with $p \geq 64$,
  $n=200$. Each box in the boxplot summarizes the difference in the
  test set prediction error (scaled RMSE) between the \mrash{}
  predictions and the ordinary least squares (OLS) estimates across 20
  simulations. The horizontal line inside each box depicts the median;
  the dot depicts the mean; the upper and lower lines depict the
  interquartile range.}
\label{fig:mrash-vs-ols}
\end{figure}

Although the paper focusses on large-scale multiple linear regression
with many predictor variables, \mrash{} can also be applied in
settings with $p \ll n$ where one would expect ordinary least squares
(OLS) to work well. To illustrate this, we simulated data sets with
PVE = 0.5, $n = 200$, $p \leq 64$ and $s = p$ (all predictors had
non-zero coefficients). The results show that \mrash{} performs
similarly to OLS when $p$ is very small and outperforms the OLS estimate
as $p$ increases (Figure~\ref{fig:mrash-vs-ols}).

\subsection{Impact of Initialization and Update Order on 
  Prediction Accuracy}
\label{subsec:initorder}

\begin{figure}[t]
\centering
\includegraphics[width=\textwidth]{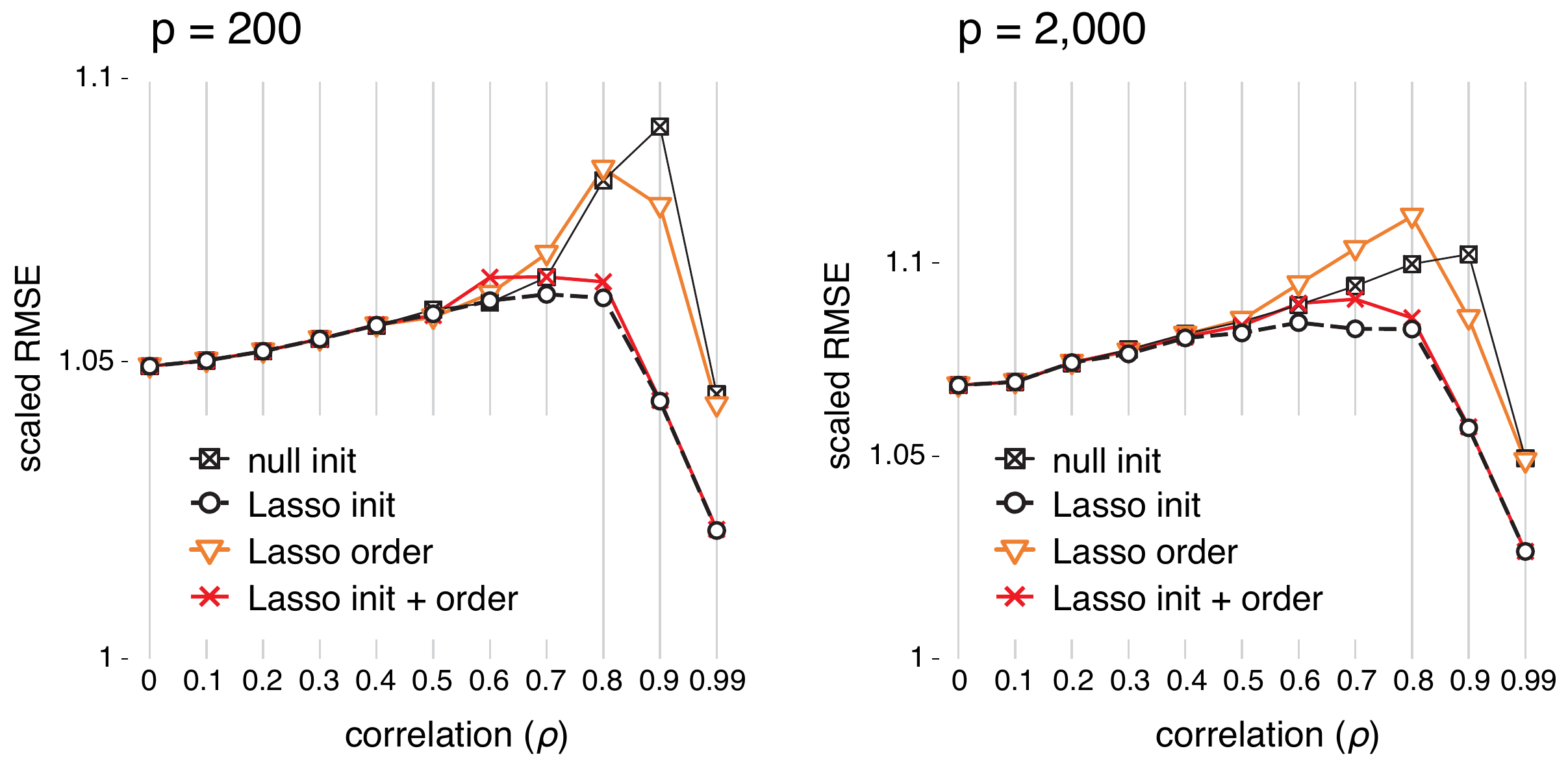}
\caption{Comparison of \mrash{} with different initializations and
  update orders. Each point shows the prediction error (scaled RMSE)
  averaged over the 20 simulations at that setting.}
\label{fig:initorder}
\end{figure}

Since \mrash{} is solving a nonconvex optimization problem, and
therefore is only guaranteed to converge to a local optimum (except in
special cases), the quality of the solution---and hence the accuracy
of the predictions---can be sensitive to initialization. This
situation is similar to other methods such as \scad{} that solve
nonconvex optimization problems, but different from methods such as
the \lasso{} that solve a convex optimization problem and are
therefore guaranteed to end up with the same final estimates
irrespective of initialization (provided of course that the algorithm
is given enough time to converge to the solution). Additionally, the
order in which the coordinatewise updates are performed can also
affect which local solution the \mrash{} algorithm converges to
\citep{ray2019variational}.  For these nonconvex optimization
problems, sometimes a ``smart'' initialization or update order can
lead to a better local solution.  From experience, we have found that
initializing \mrash{} to the cross-validated \lasso{} estimate of $\b$
seems to work well and does not greatly increase computational
effort. In this experiment, we investigated the benefits of a smart
initialization.

Specifically, we compared the following four \mrash{} variants:
\begin{itemize}

\item {\bf ``Null'' initialization.} The posterior mean coefficients
  $\bbar$ are initialized to zero and the coordinates $j = 1, \ldots,
  p$ are updated in a random order. By ``random order'', we mean a
  random permutation of the indices $1, 2, \ldots, p$. A new random
  permutation is generated for each iteration of the algorithm (that
  is, for each iteration of the repeat-until loop in
  Algorithm~\ref{alg:caisageneral}).

\item {\bf Lasso initialization.} The posterior mean coefficients are
  set to $\b = \hat{\b}^{\mathrm{lasso}}$ (see
  Section~\ref{subsec:methods-initialization}), and the coordinates $j
  = 1, \ldots, p$ are updated in a random order.

\item {\bf Lasso update order.} The coordinates $j = 1, \ldots, p$ are
  updated in the order that they are estimated to have non-zero
  coefficients as the strength of the Lasso penalty is decreased. We
  call this the ``Lasso update order,'' and it can be understood as
  the order in which the coefficients ``enter the Lasso path''
  \citep{su2017false}. (Note that determining the Lasso update order
  typically takes less effort than computing
  $\hat{\b}^{\mathrm{lasso}}$ because the cross-validation step is
  avoided.)

\item {\bf Lasso initialization and Lasso update order.}  Both the
  Lasso initialization $\bbar = \hat{\b}^{\mathrm{lasso}}$ and Lasso
  update order are used.

\end{itemize}
Initialization of $\sigma^2$ and $\pii$ is described in
Section~\ref{subsec:methods-initialization}.

To assess the benefits of these four initialization and update order
strategies, we simulated data sets with varying correlation strengths
among the predictors, then we compared the performance of the four
\mrash{} variants. The results of these simulations are summarized in
Figure~\ref{fig:initorder}. The smart initialization and update
ordering provided little benefit when the variables were not
correlated or only weakly correlated, but produced considerable gains
in prediction accuracy when the variables were strongly
correlated. Interestingly, once the coefficients were initialized to
the Lasso estimates, there was no additional benefit to updating the
coordinates using the Lasso update order.

In summary, initializing the coefficients to the cross-validated Lasso
estimates is a simple way to improve the performance of \mrash{} when
predictors are strongly correlated.

\section{More General Formulation of the Normal Means model}
\label{appendix:preliminaries}

In Section~\ref{subsec:ebnm} we defined the normal means (NM) model
for the special case when all observations $j$ have the same variance,
$\sigma^2$. This special case was sufficient to develop the VEB
methods with the assumption that $\xj^T\xj = 1$, $j = 1, \ldots,
p$. Here, we extend the NM model to allow for observation-specific
variances, $\sigma_j^2$, which is needed to generalize the VEB method
to cases in which the $\xj^T\xj = 1$ assumption no longer holds.

\subsection{The Normal Means Model}

Let $\text{NM}_p(f, {\bf s}^2)$ denote the normal means model
with prior $f$ and observation-specific variances ${\bf
  s}^2 = (s_1^2, \ldots, s_p^2) \in \mathbb{R}^p_{+}$:
\begin{equation}
\begin{aligned}
\label{eqn:general_NM}
y_j \mid b_j, s_j^2 &\sim N(b_j, s_j^2), \\
b_j &\iid f, \quad j=1,\ldots,p,
\end{aligned}
\end{equation}
such that $y_j,b_j \in \mathbb{R}$, $j = 1, \ldots, p$.
We assume priors that are mixtures of zero-mean normals, $f \in
\G(u_1^2, \dots, u_K^2)$, $u_k^2 \geq 0$, $k = 1, \ldots, K$, so that
any prior can be written as
\begin{equation*}
\bj \sim \sum_{k=1}^K \pi_k N(0, u_k^2),
\end{equation*}
such that $\pii = (\pi_1, \ldots, \pi_K) \in \mathbb{S}^K$.

As in
\eqref{eqn:latent}, it is helpful in the derivations to make
use of the latent variable representation:
\begin{equation} 
\begin{aligned}
\label{eqn:normalmixture}
p(\gamma_j = k \mid f) &= \pi_k \\
b_j \mid f, \gamma_j = k &\sim N(0, u_k^2),
\end{aligned}
\end{equation}
with $\gamma_j \in \{1, \ldots, K\}, j = 1, \ldots, p$. We write
the joint prior for $b_j, \gamma_j$ as
\begin{align} 
\label{def:joint_prior}
p_{\rm prior}(\bj,\gamma_j = k) &\triangleq 
p(\bj,\gamma_j = k \mid f) \nonumber \\ 
&= \pi_k N(b_j; 0, u_k^2).
\end{align}
In the expressions below we sometimes write the joint prior as
$p_{\mathrm{prior}}(f)$ to make its dependence on $f$ explicit.

Note that the definition of the NM model given in the main text
(Section~\ref{subsec:ebnm}), with prior \eqref{eqn:normalmixture2}, is
a special case of these definitions and can be obtained with the
substitutions $s_j^2 \leftarrow \sigma^2, j = 1, \ldots, p$ and $u_k
\leftarrow \sigma^2\sigma_k^2, k = 1, \ldots, K$.

\subsection{Posterior Distribution under Normal Means Model with One 
  Observation}

Let $q^{\rm NM}(b, \gamma \mid y, s^2, f)$ denote the posterior
distribution of $b, \gamma$ under the normal means model
$\text{NM}_1(f,s^2)$ with a single observation ($p = 1$):
\begin{equation}
\begin{aligned}
y \mid b, s^2 &\sim N(b, s^2) \\
b &\sim f.
\end{aligned}
\end{equation}
For a mixture of normals prior, $f \in \G(u_1^2, \dots, u_K^2)$, the
posterior distribution can be written as
\begin{equation} 
p^{\rm NM}(b \mid y,s^2,f) = \sum_{k = 1}^K 
\phi_{1k} \, N(b; \mu_{1k}, s_{1k}^2),
\label{eqn:optimal_form_of_q_why}
\end{equation}
in which the posterior component means $\mu_{1k}$, variances
$s_{1k}^2$, responsibilities $\phi_{1k}$, and component (marginal)
likelihoods $L_k$ are
\begin{align}
\mu_{1k} \triangleq \mu_{1k}(y;f,s^2) &= \frac{u_k^2}{s^2 + u_k^2} \times y
\label{eqn:mu_operator} \\ 
s_{1k}^2 \triangleq 
s_{1k}^2(y;f,s^2) &= \frac{s^2 u_k^2}{s^2 + u_k^2} \\
\phi_{1k} \triangleq \phi_{1k}(y;f,s^2) &=
\frac{\pi_k L_k}{\sum_{k' = 1}^K \pi_{k'} L_{k'}}
\label{eqn:phi_operator} \\
L_k \triangleq L_k(y; f, s^2) &= p(y \mid s^2, f, \gamma = k) = 
\textstyle \int p(y \mid b, s^2) \, p(b \mid f, \gamma = k) \, db 
\label{eqn:ell_operator}  
\end{align}
The posterior expressions for the NM model given in the main text
(\ref{def:posterior_mixture}--\ref{eqn:post_ebnm}) can be recovered
from these more general expressions with the substitutions $s_j^2
\leftarrow \sigma^2$, $j = 1, \ldots, p$ and $u_k \leftarrow
\sigma^2\sigma_k^2$, $k = 1, \ldots, K$.

\subsection{Evidence Lower Bound for Normal Means Model with One Observation}

Given some probability density on $b \in \mathbb{R}$, denoted by $q$,
the ELBO for the normal means model ${\rm NM}_1(f,s^2)$ with
observation $y$ is
\begin{align}
F_1^{\mathrm{NM}}(q,f,s^2; y) &= 
\log p(y \mid f,s^2) - \DKL(q \,\|\, p^{\mathrm{NM}}) \nonumber \\
&= \mathbb{E}_q[\log p(y \mid b,s^2)] - \DKL(q \,\|\, \pprior(f)) 
\nonumber \\
&= \textstyle -\frac{1}{2}\log(2{\pi}s^2) - 
\frac{1}{2s^2}\mathbb{E}_q[(y - b)^2] - \DKL(q \,\|\, \pprior(f)).
\label{eqn:nmelbo}
\end{align}
With this expression, we state the following result.

\begin{lemma}[Normal means posterior as maximum of ELBO] \label{lem:subproblem}
\rm The posterior distribution \eqref{eqn:optimal_form_of_q_why} under
the NM model $\mathrm{NM}_1(f, s^2)$ with observation $y$ maximizes
the ELBO \eqref{eqn:nmelbo}; that is,
\begin{equation*}
p^{\rm NM} = 
\argmax_q \textstyle -\frac{1}{2s^2}\mathbb{E}_q [(y - b)^2]
- \DKL(q \,\|\, p_{\rm prior}(f)).
\end{equation*}
\end{lemma}

From this lemma it follows that any $q$ maximizing the ELBO
\eqref{eqn:nmelbo} must have the following form:
\begin{equation}
q(b) = \sum_{k=1}^K \phi_{1k} N(b; \mu_{1k}, s_{1k}^2),
\label{eqn:optimal_form_of_q_oneobs}
\end{equation}
with $\phi_{1k} \geq 0$, $\mu_{1k} \in \mathbb{R}$, $s_{1k}^2 >
0$, $k = 1, \ldots, K$.

For any $q$ of the form \eqref{eqn:optimal_form_of_q_oneobs}, the ELBO
\eqref{eqn:nmelbo} has an analytic expression, which we derive in part
by making use of the formula for the K-L divergence between two normal
distributions \citep{hastie2009elements}:
\begin{equation}
F_1^{\rm NM}(q, f, s^2; y) = 
\mathbb{E}_q[\log p(y \mid b, s^2)] 
- \DKL(q \,\|\, p_{\mathrm{prior}}(f)),
\label{eqn:optimal_elbo_oneobs}
\end{equation}
in which
\begin{equation*}
\mathbb{E}_q[\log p(y \mid b, s^2)] = 
-\frac{1}{2}\log(2 \pi s^2) 
- \frac{(y - \bar{b})^2}{2s^2}
- \frac{1}{2s^2} \sum_{k=1}^K [\phi_{1k} 
(\mu_{1k}^2 + s_{1k}^2) - \bar{b}^2]
\label{eqn:expected_log_lik_nm}
\end{equation*}
and
\begin{equation*}
\DKL(q \,\|\, p_{\mathrm{prior}}(f)) = 
\sum_{k=1}^K \phi_{1k} \log\left(\frac{\phi_{1k}}{\pi_k}\right) - 
\frac{1}{2} \sum_{k = 2}^K 
\phi_{1k}\left[1 + \log\left(\frac{s_{1k}^2}{u_k^2}\right) - 
\frac{\mu_{1k}^2 + s_{1k}^2}{u_k^2} \right], 
\label{eqn::neg_kl_div_nm}
\end{equation*}
and where $\bar{b}$ is the posterior mean of $b$ with respect to $q$,
$\bar{b} = \sum_{k=1}^K \phi_{1k} \mu_{1k}$. Here we have assumed that
the first component in the prior mixture is a point mass at zero,
$\sigma_1^2 = 0$.

\subsection{ELBO for Normal Means Model with Multiple Observations}

Now we extend the above results for the single-observation NM model to
the NM model with multiple observations, $\text{NM}_p(f, {\bf s}^2)$. Since
the $b_j$'s are independent under the posterior, the ELBO is simply
the sum of the ELBOs for the single-observation NM models:
\begin{equation}
\label{eqn:elbo_nm}
F^{\rm NM}(q,f,{\bf s}^2;y) = \sum_{j=1}^p F_1^{\rm NM}(q_j,f,s_j^2;y_j).
\end{equation}
From Lemma \ref{lem:subproblem}, the $q$ that maximizes the ELBO is
\begin{align*}
q(\b) &= \prod_{j=1}^p q_j(\bj) \nonumber \\
q_j(\bj) &= p^{\mathrm{NM}}(\bj \mid y_j,f,s_j^2),
\end{align*}
It also follows that any $q$ maximizing the ELBO \eqref{eqn:elbo_nm}
must have the following form:
\begin{equation}
\begin{aligned}
q(\b) &= \prod_{j=1}^p q_j(\bj) \\
q_j(b_j) &= \sum_{k=1}^K \phi_{1jk} N(b_j; \mu_{1jk}, s_{1jk}^2),
\end{aligned}
\label{eqn:optimal_form_of_q}
\end{equation}
in which $\phi_{1jk} \geq 0$, $\mu_{1jk} \in \mathbb{R}$, $s_{1jk}^2 >
0$, $j = 1, \ldots, p$, $k = 1, \ldots, K$. For any $q$ of the form
\eqref{eqn:optimal_form_of_q}, the analytic expression for the ELBO
\eqref{eqn:elbo_nm} is easily obtained by applying the analytic
expression for the single-observation NM model
(eq.~\ref{eqn:optimal_elbo_oneobs}).

\section{Derivation of Algorithm~\protect\ref{alg:caisageneral},
  Proof of Proposition~\ref{prop:caisaderivation}}
\label{appendix:derivation}

We prove Proposition~\ref{prop:caisaderivation} by proving a slightly
more general proposition that does not require that $\xj^T \xj = 1$,
$j = 1, \ldots, p$.

\begin{proposition} \label{prop:derivation2}
\rm Let $d_j = \xj^T \xj$, $j = 1, \ldots, p$, and let 
\begin{equation*}
%\label{eqn:univariate_ols}
\tilde{b}_j \triangleq \frac{\xj^T\rbar_j}{\xj^T \xj}
\end{equation*}
denote the ordinary least squares (OLS) estimate of the coefficient
$\bj$ when the residuals $\rbar_j$ are regressed against $\xj$.
See Proposition~\ref{prop:caisaderivation} for more definitions.
Then we have the following results:
\begin{enumerate}[(i)]

\item[(i)] The coordinate ascent update $q_j^{\ast} \triangleq
  \argmax_{q_j} F(q, g, \sigma^2)$ is obtained by
\begin{equation*}
q_j^{\ast}(b_j) = p^{\mathrm{NM}}(b_j; \tilde{b}_j, \sigma^2/d_j, g_{\sigma}),
\end{equation*}
in which $p^{\mathrm{NM}}$, defined in
\eqref{eqn:optimal_form_of_q_why}, is the posterior distribution
of $b$ under the following NM model:
\begin{equation}
\label{eqn:real_subprob_model}
\begin{aligned}
\tilde{b} \mid b, \sigma^2 &\sim N(b,\sigma^2/d_j) \\
b \mid g, \sigma^2 &\sim g_{\sigma}.
\end{aligned}
\end{equation}

\item[(ii)] The coordinate ascent update
\begin{equation*}
g^{\ast} \triangleq \argmax_{g \,\in\, \Gfm} F(q, g, \sigma^2)
\end{equation*}
is achieved by setting
\begin{equation*} 
\begin{aligned}
g^{\ast} &= \textstyle \sum_{k=1}^K \pi_k^{\ast} N(0,\sigma_k^2) \\
\pi_k^{\ast} &= \textstyle \frac{1}{p} \sum_{j=1}^n q_j(\gamma_j = k), 
\quad k = 1,\ldots,K.
\end{aligned}
\end{equation*}
% In particular, if $q_j$ is updated as in (i) above, so that $q_j(b_j)
% = p^{\mathrm{NM}}(b_j; \tilde{b}_j, g, \sigma^2/d_j)$, then 
% $q_j(\gamma_j = k)$ is equal to the responsibility
% $\phi_{k}(\tilde{b}_j; g, \sigma^2/d_j)$
% \eqref{eqn:phi_operator}.

\item[(iii)] Using the parameterization of $q$ in
  \eqref{eqn:optimal_form_of_q}, and assuming that $g$ is updated as
  in (ii) above, and $\sigma_1^2 = 0$, the coordinate ascent update
\begin{equation*}
(\sigma^2)^{\ast} \triangleq \argmax_{\sigma^2 \,\in\,
\mathbb{R}_{+}} F(q, g, \sigma^2)
\end{equation*}
is achieved by setting
\begin{align*}
(\sigma^2)^{\ast}
%  &= \frac{\|\rbar\|^2
% + \sum_{j=1}^p \mathrm{Var}_q(b_j) + 
% \sum_{j=1}^p \sum_{k=2}^K \mathbb{E}[b_j \mid \gamma_k = k]/\sigma_k^2}
% {n + p(1-\pi_1^{\ast})} \\
&= \frac{\|\rbar\|^2
+ \sum_{j=1}^p \sum_{k=2}^K \phi_{jk} 
(d_j + 1/\sigma_k^2) (\mu_{1jk}^2 + s_{1jk}^2) 
- \sum_{j=1}^p d_j \bar{b}_j^2}{n + p(1-\pi_1^{\ast})}.
\end{align*}
Additionally, if $q_1, \ldots, q_p$ are updated as in (i) above,
we obtain the simpler expression
\begin{equation}
\label{eqn:new_update_sigma2}
(\sigma^2)^{\ast} = \frac{\|\rbar\|^2
+ \bbar^T{\bf D}(\tilde{\b} - \bbar) + \sigma^2p(1 - \pi_1^{\ast})}
{n + p(1-\pi_1)},
\end{equation}
where ${\bf D}$ is the $p \times p$ diagonal matrix with diagonal
entries $d_1, \ldots, d_p$.
\end{enumerate}
Note that Proposition \ref{prop:caisaderivation} is a special case of
this proposition when $d_j = 1$, for $j = 1, \ldots, p$.
\end{proposition}

In the next sections, we prove parts (i), (ii) and (iii) of
Proposition~\ref{prop:derivation2}. These proofs start from the ELBO
\eqref{eq:elbo}.
% which for convenience, we reproduce here:
% \begin{equation*}
% F(q,g,\sigma^2) \triangleq 
% \log p(\y \mid \X, g, \sigma^2) - \DKL(q \,\|\, p_{\rm post})
% \end{equation*}
From Bayes' rule,
\begin{equation*}
p_{\rm post}(\b) = 
\frac{p(\y \mid \X, \b, \sigma^2) \, p(\b \mid g, \sigma^2)}
     {\int p(\y \mid \X, \b, \sigma^2) \, p(\b \mid g, \sigma^2) \, d{\bf b}},
\end{equation*}
we can write the ELBO as
\begin{equation}
F(q,g,\sigma^2) = 
\mathbb{E}_q[\log p(\y \mid \X, \b, \sigma^2)]
- \sum_{j=1}^p \DKL(q_j \,\|\, \pprior).
\label{eqn:elbo_bayes_rule}
\end{equation}
Next, using the property that $q$ factorizes over the individual
coordinates $j = 1, \ldots, p$, we have
\begin{equation*}
\label{eqn:elbo_appendixa}
F(q, g, \sigma^2) = 
-\frac{n}{2} \log(2\pi\sigma^2) 
- \frac{1}{2\sigma^2}\mathbb{E}_q[\norm{\y - \X\b}^2] 
- \sum_{j=1}^p \DKL(q_j \,\|\, \pprior).
\end{equation*}

\subsection{Update for $q_j$}
\label{appendix:updateofq}

The coordinate ascent update for $q_j$ involves solving the following
optimization problem:
\begin{equation*}
%\label{eqn:update_of_qj_abstract}
q_j^{\ast} = \argmax_{q_j} F(q, g, \sigma^2).
\end{equation*}
From \eqref{eqn:elbo_bayes_rule}, this is
equivalent to solving
\begin{equation*}
q_j^{\ast} = \argmax_{q_j} \mathbb{E}_q[\log p(\y \mid \X, \b, \sigma^2)] - 
\DKL(q_j \,\|\, p_{\mathrm{prior}})
\end{equation*}
By rearranging terms, it can be shown that this is equivalent to
solving
\begin{equation}
q_j^{\ast} = \argmax_{q_j} -\frac{d_j}{2\sigma^2} 
\mathbb{E}_{q_j}[(\tilde{b}_j - \bj)^2] 
- \DKL(q_j \,\|\, p_{\mathrm{prior}}).
\label{eqn:conjugacy_subprob}
\end{equation}
The right-hand side of \eqref{eqn:conjugacy_subprob} is the ELBO for
the NM model \eqref{eqn:real_subprob_model}; that is, if we ignore
constant terms, the ELBO in \eqref{eqn:nmelbo} recovers
\eqref{eqn:conjugacy_subprob} by making the substitutions $y
\leftarrow \tilde{b}_j$, $s^2 \leftarrow \sigma^2/d_j$, $f \leftarrow
g_{\sigma}$. And, therefore, from Lemma~\ref{lem:subproblem}---and
specifically from \eqref{eqn:optimal_form_of_q_oneobs}---we have
\begin{equation*}
q_j^{\ast}(b_j) =
p^{\mathrm{NM}}(b_j \mid \tilde{b}_j, \sigma^2/d_j, g_{\sigma}) 
= \sum_{k=1}^K \phi_{1jk} N(b_j; \mu_{1jk}, s_{1jk}^2),
%\label{eqn:optimal_form}
\end{equation*}
in which
\begin{align*}
\phi_{1jk} &= \phi_{1k}(\tilde{b}_j, g_{\sigma}, \sigma^2/d_j) \\
\mu_{1jk} &= \mu_{1k}(\tilde{b}_j, g_{\sigma}, \sigma^2/d_j) \\
s_{1jk}^2 &= s_{1k}^2(\tilde{b}_j, g_{\sigma}, \sigma^2/d_j).
\end{align*}
This proves part (i) of Proposition~\ref{prop:derivation2}.

\subsection{Update for $g$}
\label{appendix:updateofg}

The coordinate ascent update for $g$ involves solving the following
optimization problem:
\begin{equation} 
\label{eqn:update_of_g_abstract}
g^{\ast} \leftarrow \argmax_{g \,\in\, \mathcal{G}} F(\q, g, \sigma^2).
\end{equation}
Recall, for the mixture prior with fixed mixture components, fitting
$g$ reduces to fitting the mixture weights, $\pii$. Since $\pii$ only
appears in the ELBO in the K-L divergence term with respect to the
prior, solving \eqref{eqn:update_of_g_abstract} is equivalent to
solving
\begin{equation*}
\pii^{\ast} = \argmin_{\pii \,\in\, \mathbb{S}^K}
\sum_{j=1}^p \DKL(q_j \,\|\, p_{\mathrm{prior}}),
\end{equation*}
which simplifies further:
\begin{equation*}
\pii^{\ast} = \argmax_{\pii \,\in\, \mathbb{S}^K}
\sum_{j=1}^p \sum_{k=1}^K \phi_{jk} \log \pi_k.
\end{equation*}
This has the following analytic solution:
\begin{equation*}
\pi_k^{\ast} = \frac{1}{p} \sum_{j=1}^p \phi_{jk}, \quad k = 1, \ldots, K.
\end{equation*}
This proves part (ii) of Proposition~\ref{prop:derivation2}.

This update can be thought of as an approximate M-step update for the
mixture weights in which the posterior probabilities (the
``responsibilities'') are computed approximately using $q$.

\subsection{Update for $\sigma^2$}
\label{appendix:updateofsigma}

The coordinate ascent update for the residual variance $\sigma^2$ is
the solution to
\begin{equation*}
%\label{eqn:update_of_sigma_abstract}
(\sigma^2)^{\ast} = \argmax_{\sigma^2 \,\in\, \mathbb{R}_{+}} 
F(\q, g, \sigma^2).
\end{equation*}
From earlier results, the ELBO for any $q$ parameterized as
\eqref{eqn:optimal_form_of_q} works out to
\begin{equation}
F(q, g, \sigma^2) =
\mathbb{E}_q[\log p(\y \mid \X, \b, \sigma^2)]
- \sum_{j=1}^p \DKL(q_j \,\|\, p_{\mathrm{prior}})
\label{eqn:elbo_parametrized}
\end{equation}
in which
\begin{align*}
\mathbb{E}_q[\log p(\y \mid \X, \b, \sigma^2)] &=
-\frac{n}{2} \log (2\pi\sigma^2) 
- \frac{1}{2\sigma^2} \left\|\y - \X\bbar \right\|^2 
 - \frac{1}{2\sigma^2} \sum_{j=1}^p d_j 
\left[\sum_{k=1}^K \phi_{1jk} (\mu_{1jk}^2 + s_{1jk}^2) -
  \bar{b}_j^2 \right]
% \label{eqn:expected_log_lik}
\\
\DKL(q_j \,\|\, p_{\mathrm{prior}}) &= 
\sum_{k=1}^K \phi_{1jk} \log \frac{\phi_{1jk}}{\pi_k} 
- \frac{1}{2} \sum_{k=2}^K \phi_{1jk}
\left[1 + \log\frac{s_{1jk}^2}{\sigma^2 \sigma_k^2} 
- \frac{s_{1jk}^2 + \mu_{1jk}^2}{\sigma^2 \sigma_k^2} \right], 
% \label{eqn:neg_kl_div}
\end{align*}
and where $\bar{b}_j = \sum_{k=1}^K \phi_{1jk} \mu_{1jk}$. Taking the
partial derivative of $F$ with respect to $\sigma^2$ then solving for
$\sigma^2$ yields the following update:
\begin{equation*}
%\label{eqn:update_sigma2_param}
(\sigma^2)^{\ast} = 
\frac{\norm{\rbar}^2 
+ \sum_{j=1}^p \sum_{k=2}^K \phi_{1jk} (d_j + 1/\sigma_k^2) 
(\mu_{1jk}^2 + s_{1jk}^2) 
- \sum_{j=1}^p d_j \bar{b}_j^2}{n + p(1 - \pi_1)}.
\end{equation*}
When $\sigma^2$ is updated following updates to $q$, we can
simplify this expression by noting the specific form of the posterior
means and variances,
\begin{align*}
s_{1jk}^2 &= \frac{\sigma^2}{d_j + 1/\sigma_k^2} \\
\mu_{1jk} &= \frac{d_j}{d_j + 1/\sigma_k^2} \times \tilde{b}_j,
\end{align*}
for $k = 2, \ldots, K$, which gives
\begin{align*}
(\sigma^2)^{\ast}
% &= 
% \frac{\norm{\y - \X\bbar}^2 
% + \sum_{j=1}^p \sum_{k=2}^K \phi_{1jk}
% (\mu_{1jk} d_j \tilde{b}_j + \sigma^2) 
% - \sum_{j=1}^p d_j \bar{b}_j^2}
% {n + p(1-\pi_1)} \\
&= \frac{\norm{\bar{\bf r}}^2 
+ \sum_{j=1}^p d_j \bar{b}_j (\tilde{b}_j - \bar{b}_j) 
+ \sigma^2 p(1 - \pi_1)}{n + p(1-\pi_1)}.
\end{align*}
This proves part (iii) of Proposition~\ref{prop:derivation2}.

\subsection{More detailed VEB algorithm}
\label{appendix:actualalg}

Further details about the implementation of the VEB algorithm are
given in Algorithm~\ref{alg:actualalgorithm}.
Algorithm~\ref{alg:actualalgorithm}
does not require $\xj^T\xj = 1$, $j = 1, \ldots, p$.

\begin{algorithm}[htbp]
\caption{Coordinate ascent for fitting VEB model (more detailed).}
\label{alg:actualalgorithm}
\begin{algorithmic}
\REQUIRE Data ${\bf X} \in \mathbb{R}^{n \times p}, {\bf y} \in
  \mathbb{R}^n$; number of mixture components, $K$; \\ prior variances,
  $\sigma_1^2 < \cdots < \sigma_K^2$, with $\sigma_1^2 = 0$; initial
  estimates $\bar{\bf b}, {\bm\pi}, \sigma^2$.
\STATE $\rbar \leftarrow \y - \X\bbar$ \hfill (compute mean residuals)%
\STATE $t \leftarrow 0$
\FOR{$j \leftarrow 1 \mbox{ to } p$}
\STATE $d_j = \xj^T\xj$
\ENDFOR
\REPEAT
\FOR{$j \leftarrow 1 \mbox{ to } p$}
\STATE $\bar{\mathbf{r}}_j = \bar{\mathbf{r}} + \xj\bar{b}_j$
\hfill (disregard $j$th effect in residuals)%
\STATE $\tilde{b}_j \leftarrow \xj^T\bar{\bf r}_j/d_j$
\hfill (compute OLS estimate)%
\FOR{$k \leftarrow 1 \mbox{ to } K$}
\STATE $\displaystyle \mu_{jk} \leftarrow \frac{d_j}{d_j + 1/\sigma_k^2} 
        \times \tilde{b}_j$
\hfill (update $q_j$)%
\STATE $\displaystyle \phi_{jk} \leftarrow \pi_k \times
  \frac{1}{1 + d_j \sigma_k^2} \times
  \exp\bigg(\frac{d_j \tilde{b}_j \mu_{jk}}{2\sigma^2}\bigg)$ 
\ENDFOR
\FOR{$k \leftarrow 1 \mbox{ to } K$}
\STATE $\phi_{jk} \leftarrow \phi_{jk}/\sum_{k'=1}^K \phi_{jk'}$
\ENDFOR
\STATE $\bar{b}_j \leftarrow \sum_{k=1}^K \phi_{jk} \mu_{jk}$
\hfill (update posterior mean of $b_j$)%
\STATE $\bar{\mathbf{r}} \leftarrow \bar{\mathbf{r}}_j - \xj\bar{b}_j$.
\hfill (update mean residuals)%
\ENDFOR
\FOR{$k \leftarrow 1 \mbox{ to } K$}
\STATE $\pi_k \leftarrow \sum_{j=1}^p \phi_{jk}/p$.
\hfill (update $g$; eq.~\ref{eqn:g_update})%
\ENDFOR
\STATE $\displaystyle
(\sigma^2)^{\ast} \leftarrow \frac{\|\rbar\|^2
+ \bbar^T{\bf D}(\tilde{\b} - \bbar) + \sigma^2p(1 - \pi_1^{\ast})}
{n + p(1-\pi_1)}$
\hfill (update $\sigma^2$; eq.~\ref{eqn:new_update_sigma2})%
\STATE $t \leftarrow t + 1$.
\UNTIL{convergence criterion is met}
\RETURN $\bbar, \pii, \sigma^2$
\end{algorithmic}
\end{algorithm}

\section{VEB as a PLR}
\label{appendix:plr}

\subsection{Normal Means as a Penalized Estimation Problem}

In this section, we formulate the normal means problem with one
observation, ${\mathrm{NM}}_1(f, s^2)$, as penalized estimation
problem. Specifically, we express the posterior mean of $b$ under the
NM model as a solution to a penalized least squares problem.  The
results for the single-observation NM model are used later to derive
results for the multiple linear regression model.

Recall, $\FNM_1(q, f, s^2; y)$ denotes the ELBO for the
single-observation NM model \eqref{eqn:nmelbo}. From this, we
define
\begin{equation}
\label{eqn:fnm_b}
h_1^{\mathrm{NM}}(\bar{b}, f, s^2; y) \triangleq
-\max_{q\,:\, \mathbb{E}_q(b) \,=\, \bar{b}} \; \FNM_1(q,f,s^2; y).
\end{equation}
As a reminder, $\FNM_1(q, f, s^2; y)$ attains its maximum over $q$ the
exact posterior, $q = p^{\mathrm{NM}}$; analogously,
$h_1^{\mathrm{NM}}(\bar{b}, f, s^2; y)$ attains its minimum over
$\bar{b}$ at $\bar{b} = S_{f,s}(y)$, the posterior mean of $b$ (see
Definition \ref{def:posterior_mean_shrinkage}). Further, at their
respective optima these two functions recover the marginal
likelihood:
\begin{align*}
\log p(y \mid f, s^2) &= F_1^{\mathrm{NM}}(p^{\mathrm{NM}}, f,s^2;y) \\
&= \textstyle \max_q F_1^{\mathrm{NM}}(q, f, s^2; y) \\
&= \textstyle \max_{\bar{b} \,\in\, \mathbb{R}} 
   \max_{q \,:\, \mathbb{E}_q(b)\,=\,\bar{b}} 
   F_1^{\mathrm{NM}}(q, f, s^2; y) \\
&= \textstyle \max_{\bar{b} \,\in\, \mathbb{R}} 
-h_1^{\mathrm{NM}}(\bar{b}, f, s^2; y) \\
&= -h_1^{\mathrm{NM}}(S_{f,s}(y), f, s^2; y).
\end{align*}
With these definitions, we can express the posterior mean for $b$ as
the solution to a real-valued optimization problem:
\begin{equation*}
S_{f,s}(y) = \textstyle \argmin_{\bar{b}} \,
h_1^{\mathrm{NM}}(\bar{b}, f, s^2; y).
\end{equation*}
The following lemma states that this can be understood as optimizing a
penalized loss function and gives an explicit form for the penalty.

\begin{lemma}
\label{lem:nm_as_plr}
\rm $h_1^{\mathrm{NM}}(\bar{b}, f, s^2; y)$ can be written as a
penalized loss function,
\begin{equation}
\label{eqn:fnm_penalized}
h_1^{\mathrm{NM}}(\bar{b}, f, s^2; y) =
\frac{1}{2s^2}(y - \bar{b})^2
+ \frac{1}{s^2}\rho_{f,s}(\bar{b}),
\end{equation}
in which the penalty is
\begin{equation}
\label{eqn:pendef}
\rho_{f, s}(\bar{b}) \triangleq
\min_{q\,:\,\mathbb{E}_q(b) \,=\, \bar{b}} \,
\frac{s^2}{2} \log(2\pi s^2) +
\frac{1}{2}\mathrm{Var}_q(b) +
s^2\DKL(q \,\|\, \pprior(f)).
\end{equation}
For any $y \in \mathbb{R}$, this penalty term satisfies
\begin{equation}
\label{eqn:pensat}
\rho_{f,s}(S_{f,s}(y)) = \textstyle
-s^2\ell_{\mathrm{NM}}(y; f, s^2) - \frac{1}{2}(y - S_{f,s}(y))^2,
\end{equation}
and
\begin{align}
\label{eqn:pensat-diff}
\rho_{f,s}'(S_{f,s}(y)) &= (y - S_{f,s}(y)) \\
\label{eqn:tweedie}
&= - s^2 \ell_{\mathrm{NM}}'(y; f,s^2),
\end{align}
in which $\ell_{\mathrm{NM}}(y; f, s^2)$ is the marginal
log-likelihood $\ell_{\mathrm{NM}}(y; f, s^2) \triangleq \log p(y \mid
f, s^2)$ for the single-observation normal means model,
$\mathrm{NM}_1(f, s^2)$.
\end{lemma}

\begin{proof}
From \eqref{eqn:nmelbo}, we have
\begin{align*}
\FNM_1(q,f,s^2;y) 
% &= -\frac{1}{2} \log(2 \pi s^2) -
% \frac{1}{2s^2}\mathbb{E}_q[(y - b)^2]
% - \DKL(q \,\|\, \pprior(f)) \\
&=
-\frac{1}{2s^2}(y - \mathbb{E}_q(b))^2
- \bigg[ \frac{1}{2} \log(2\pi s^2) +
         \frac{1}{2s^2}\,\mathrm{Var}_q(b) +
         \DKL(q \,\|\, \pprior(f))
         \bigg].
\end{align*}

Expressions \eqref{eqn:fnm_penalized} and \eqref{eqn:pendef} follow
from \eqref{eqn:fnm_b}.

Expression \eqref{eqn:pensat} is obtained by substituting $\bar{b} =
S_{f,s}(y)$ into \eqref{eqn:fnm_penalized} and rearranging, noting
that $h(\bar{b}, f, s^2; y)$ attains its minimum at this $\bar{b}$ and
therefore recovers the marginal log-likelihood,
\begin{equation*}
h^{\mathrm{NM}}(S_{f,s}(y),f,s^2;y) 
= -\ell_{\mathrm{NM}}(y; f, s^2).
\end{equation*}

Expression \eqref{eqn:pensat-diff} is a consequence of the fact that
$h^{\mathrm{NM}}(\bar{b}, f, s^2; y)$ attains its minimum at $\bar{b} =
S_{f,s}(y)$ for any $y \in \mathbb{R}$, that is,
\begin{equation*}
S_{f,s}(y) = \argmin_{\bar{b} \,\in\, \mathbb{R}} \;
h^{\mathrm{NM}}(\bar{b}, f, s^2; y).
\end{equation*}
Therefore, the derivative of \eqref{eqn:fnm_penalized} with respect to
$\bar{b}$ at $\bar{b} = S_{f,s}(y)$ must be zero. Finally,
\eqref{eqn:tweedie} is obtained by applying Tweedie's formula
\citep{efrontweedie}.
\end{proof}

\subsection{VEB as a Penalized Regression Problem}

Here we consider the ELBO for the multiple linear regression model
\eqref{eq:elbo}.
% \begin{equation}
% \label{eqn:def_elbo_bbar}
% h(\bar{\b}, g, \sigma^2) \triangleq 
% \max_{q\,:\, \mathbb{E}_q(\b) \,=\, \bbar} \; F(q, g, \sigma^2).
% \end{equation}
We begin with the following lemma.
\begin{lemma} \label{lem:enorm}
\rm If the distribution $q(\b)$ factorizes as $q(\b) = \prod_{j=1}^p
q_j(b_j)$, then
\begin{equation*}
\mathbb{E}_q[\norm{{\bf r}}^2] = 
\norm{\rbar}^2 + \sum_{j=1}^p d_j \mathrm{Var}_{q_j}(b_j),
\end{equation*}
where $\bbar \triangleq \mathbb{E}_q[\b]$, ${\bf r} \triangleq \y -
\X\b$ and ${\bf r} \triangleq \mathbb{E}_q[{\bf r}] = \y - \X\bbar$.
\end{lemma}
\begin{proof}
\begin{align*}
\mathbb{E}_q[\norm{{\bf r}}^2]
&= \mathbb{E}_q[\norm{\rbar + \X(\bbar - \b)}^2] \\
&= \norm{\rbar}^2 + \mathbb{E}_q[\norm{\X(\bbar - \b)}^2] \\
&= \norm{\rbar}^2 + 
\mathbb{E}_q[ (\bbar - \b)^T \X^T\X(\bbar - \b)] \\
&= \norm{\rbar}^2 + \tr{\X^T\X\mathrm{Cov}_q(\b)} \\
& = \textstyle \norm{\rbar}^2 + \sum_{j=1}^p d_j {\rm Var}_{q_j}(b_j).
\end{align*}
\end{proof}
In the following proposition, we express the ELBO for the multiple
linear regression model as a penalized loss function.
\begin{proposition} 
\label{prop:elbo_as_plr_appendix}
\rm The objective function $h$ \eqref{eqn:define_elbo_as_plr} can be
written as a penalized loss function,
\begin{equation}
\label{eqn:elbo_as_plr_appendix}
h(\bar{\b}, g, \sigma^2) = 
\frac{1}{2\sigma^2} \norm{\y - \X\bbar}^2 
+ \sum_{j=1}^p \rho_{{g_\sigma},s_j}(\bar{b}_j)/s_j^2
+ \frac{1}{2} \sum_{j=1}^p \log(d_j)
+ \frac{n-p}{2} \log(2\pi\sigma^2),
\end{equation}
using the penalty function $\rho_{f,s}$ defined in \eqref{eqn:pendef},
and defining $s_j^2 \triangleq \sigma^2/d_j$, $j = 1, \ldots, p$. Note
that when $d_j = 1$, $k = 1, \ldots, p$,
\eqref{eqn:elbo_as_plr_appendix} simplifies to
\eqref{eqn:elbo_as_plr_full_expression}.
\end{proposition}

\begin{proof}
From Lemma~\ref{lem:enorm}, we have
\begin{align*}
h(\bar{\b}, g, \sigma^2) &= 
\frac{n}{2} \log(2\pi\sigma^2) 
+ \frac{1}{2\sigma^2}\mathbb{E}_q[\norm{\y - \X\b}^2]  
+ \sum_{j=1}^p \DKL(q_j \,\|\, \pprior(g_\sigma)) \\
&= \frac{n}{2} \log(2\pi\sigma^2) 
+ \frac{1}{2\sigma^2} \norm{\rbar}^2 
+ \frac{1}{2} \sum_{j=1}^p \mathrm{Var}_{q_j}(b_j)/s_j^2
+ \sum_{j=1}^p \DKL(q_j \,\|\, \pprior(g_{\sigma})).
\end{align*}
Therefore,
\begin{align*}
h(\bar{\b}, g, \sigma^2) &= \max_{q\,:\, \mathbb{E}_q(\b) \,=\, \bbar} 
\; F(q, g, \sigma^2) \\
&= \frac{n}{2} \log (2\pi\sigma^2) 
+ \frac{1}{2\sigma^2} \norm{\rbar}^2 \\
& \qquad + \sum_{j=1}^p  
\frac{1}{s_j^2} \times \bigg\{
\min_{q_j\,:\, \mathbb{E}_{q_j}(b_j) \,=\, \bar{b}_j} 
\textstyle \frac{1}{2}\mathrm{Var}_{q_j}(b_j)
+ s_j^2\DKL(q_j \,\|\, \pprior(g_{\sigma})) \bigg\} \\
&= \frac{n}{2} \log (2\pi\sigma^2) +
\frac{1}{2\sigma^2} \norm{\rbar}^2 
+ \sum_{j=1}^p \frac{1}{s_j^2}  
\bigg[\rho_{{g_\sigma},s_j}(\bar{b}_j) 
- \frac{s_j^2}{2}\log(2\pi s_j^2) \bigg] \\
&= \frac{1}{2\sigma^2} \norm{\rbar}^2 
+ \sum_{j=1}^p \rho_{{g_{\sigma}}, s_j}(\bar{b}_j)/s_j^2 
+ \frac{1}{2} \sum_{j=1}^p \log(d_j)
+ \frac{n-p}{2} \log (2\pi\sigma^2).
\end{align*}
\end{proof}

\section{Additional Results and Proofs}
\label{appendix:proof}

\begin{proposition}[Convergence of cyclic coordinate ascent for VEB] 
\label{prop:conv}
\rm The sequence of iterates
\begin{equation*}
\{q^{(t)}, g^{(t)}, (\sigma^2)^{(t)}\}, \quad t = 0, 1, 2, \ldots,
\end{equation*}
generated by Algorithm~\ref{alg:caisageneral} converge monotonically
to a stationary point of the ELBO, $F$ \eqref{eq:elbo}.
\end{proposition}
\begin{proof}
By Proposition~2.7.1 of \citet{Bertsekas9}, the sequence of iterates
$\{q^{(t)}, g^{(t)}, (\sigma^2)^{(t)}\}$, $t = 0, 1, 2, \ldots$,
generated by Algorithm \ref{alg:caisageneral} converges monotonically
to a stationary point of $F$ provided that $F$ is continuously
differentiable and each coordinate update,
\begin{align*}
q_j^{(t+1)} &= \mathrm{argmax}_{q_j} \, 
F(q_1^{(t+1)},\cdots,q_{j-1}^{(t+1)},q_j,q_{j+1}^{(t)},\cdots,q_p^{(t)}, g
^{(t)},(\sigma^2)^{(t)}), \quad j = 1, \ldots, p \\
g^{(t+1)} &= \mathrm{argmax}_{g \,\in\, \mathcal{G}} \,
F(q^{(t+1)}, g, (\sigma^2)^{(t)}) \\
(\sigma^2)^{(t+1)} &= \mathrm{argmax}_{\sigma^2 \,\in\, \mathbb{R}_{+}} \,
F(q^{(t+1)}, g^{(t+1)}, \sigma^2),
\end{align*}
is finite and uniquely determined. (See also
\citealt{luo1992convergence, tseng2001convergence} for another
treatment of convergence of coordinate ascent under general
conditions.) A sufficient condition for $F$ to be continuously
differentiable and for the coordinate ascent updates (Proposition
\ref{prop:caisaderivation} or Proposition \ref{prop:derivation2}) to
have a unique solution is that $0 < \sigma^2 < \infty$, $\pi_k > 0$
for all $k = 1, \ldots, K$, and $0 \leq \sigma_1^2 < \cdots <
\sigma_K^2 < \infty$.
\end{proof}

\subsection{Proof of Proposition~\ref{prop:orthogonal}}

The ELBO
\begin{equation*}
F(q,g,\sigma^2) = 
\iint q(\b, {\bm\gamma})
\log \bigg\{
\frac{p(\y \mid \X, \b, \sigma^2) \, p(\b, {\bm\gamma} \mid g, \sigma^2)}
     {q(\b, {\bm\gamma})} \bigg\} \, d\b \, d{\bm\gamma}
\end{equation*}
is maximized with respect to $q$ when $q(\b, {\bm\gamma}) \propto p(\y
\mid \X, \b, \sigma^2) \, p(\b, {\bm\gamma} \mid g, \sigma^2)$. This
follows from the equality condition of Jensen's inequality
\citep{jordan1999introduction}. When the columns of $\X$ are
orthogonal, the posterior factorizes over the individual coordinates $j$,
\begin{align*}
p(\b, {\bm\gamma} \mid \X, \y, g, \sigma^2) &\propto
p(\y \mid \X, \b, \sigma^2) \,
p(\b, {\bm\gamma} \mid g, \sigma^2) \\
& \propto \prod_{j=1}^p
\exp\bigg\{-\frac{(b_j - \xj^T \y)^2}{2\sigma^2}\bigg\}
\times p(b_j,\gamma_j \mid g, \sigma^2).
\end{align*}
Therefore, when $\X$ has orthogonal columns, the best $q$, even with
the restriction of being fully factorized \eqref{def:mean-field}, is
able to recover the exact posterior since the exact posterior also
factorizes over the coordinates $j = 1, \ldots, p$.

\subsection{Proof of Proposition~\ref{prop:vpm_solves_an_opt}}

First, we note that
\begin{equation*}
F(\qhat, \ghat, \shat) =
\max_{q \,\in\, \mathcal{Q}} F(q, \ghat, \shat) =
-h(\hat{\b}, \ghat, \shat).
\end{equation*}
Hence, for any $\bbar \in \mathbb{R}^p$, we have
\begin{align*}
-h(\bbar, g, \sigma^2) &= \textstyle
\max_{q \,\in\, \Q, \,\mathbb{E}_q[\b] \,=\, \bbar} \; F(q,g,\sigma^2) \\
&\leq \textstyle \max_{q \,\in\, \mathcal{Q}} \; F(q,g,\sigma^2) \\
&\leq \textstyle \max_{q \,\in\, \mathcal{Q},\,g \,\in\, \mathcal{G},
            \,\sigma^2 \,\in\, \mathcal{T}} \; F(q, g, \sigma^2) \\
&= F(\qhat, \ghat, \shat) \\
&= -h(\hat{\b}, \ghat, \shat).
\end{align*}
This proves that 
\begin{equation*}
\hat{\b}, \ghat, \shat =
\argmin_{\bbar \,\in\, \mathbb{R}^p,\, g \,\in\, \mathcal{G},\,
  \sigma^2 \,\in\, \mathcal{T}} \; h(\bbar, g, \sigma^2).
\end{equation*}

\subsection{Proof of Theorem~\ref{thm:veb_penloglik}}

The proof of the first part of Theorem~\ref{thm:veb_penloglik} follows
immediately from Proposition~\ref{prop:elbo_as_plr_appendix} and Lemma
\ref{lem:nm_as_plr} by letting $d_j = 1$ for $j = 1, \ldots, p$.

The missing piece of the proof is to show that $S_{\rho_{f,\sigma}}(y)
= S_{f,\sigma}(y)$. We start with the definition of $S_{\rho}$ in
\eqref{def:shrinkage_thresholding_operator},
\begin{equation*} 
S_{\rho_{f, \sigma}}(y) = 
\argmin_{b \,\in\, \mathbb{R}}\frac{1}{2}(y - b)^2 + \rho_{f,\sigma}(b).
\end{equation*}
To solve for the argmin on the right-hand side, we differentiate with
respect to $b$ and set the derivative to zero,
\begin{equation*}
\rho_{f,\sigma}'(b) = y - b.
\end{equation*}
From \eqref{eqn:penderiv}, this derivative vanishes when $b =
S_{f,\sigma}(y)$. Therefore, $S_{\rho_{f,\sigma}}(y) =
S_{f,\sigma}(y)$.

\subsection{Mathematical Properties of the Posterior Mean Shrinkage Operator}

\begin{lemma}
\label{lem:property_pm_shrinkage}
\rm Let $f$ be a symmetric unimodal distribution on $\mathbb{R}$ with
a mode at zero, and assume $\sigma^2 > 0$. Then the NM posterior mean
operator $S_{f, \sigma}(y)$ defined in \eqref{eqn:shrink} is
symmetric, non-negative, and non-decreasing on $y \in \mathbb{R}$, and
$S_{f, \sigma}(y) \leq y$ on $y \in (0, \infty)$. That is, $S_{f,
  \sigma}$ is a ``shrinkage operator'' that shrinks towards zero.
\end{lemma}

\begin{proof}
% For convenience, we restate the NM model here, originally defined in
% \eqref{eqn:nmf}:
% \begin{equation*}
% \begin{aligned}
% y &\mid b, \sigma^2 \sim N(b, \sigma^2) \\
% b &\sim f.
% \end{aligned}
% \end{equation*}
The marginal likelihood for the NM model \eqref{eqn:nmf} is
\begin{align*}
\ell_{\mathrm{NM}} (y; f, \sigma^2) 
&\triangleq \log p(y \mid f, \sigma^2) \\
&= \textstyle \int p(y \mid b, \sigma^2) \, p(b) \, db \\
&= \textstyle \int N(y; b, \sigma^2) \, f(b) \, db.
\end{align*}
By Khintchine’s representation theorem
\citep{dharmadhikari1988unimodality}, $f$ can be represented as
mixture of uniform distributions,
\begin{equation*}
f(b) = \int_0^\infty \frac{\mathbb{I}\{|b| < t\}}{2t} \, p(t) \, dt
\end{equation*}
for some (possibly improper) univariate mixing density $p(t)$. Let
$p(b \mid t)$ be the density function of the uniform distribution
on $[-t, t]$. Then we have
\begin{align}
p(y \mid f, \sigma^2) 
&= \textstyle \int p(y \mid b,\sigma^2) \, p(b \mid f, \sigma^2) \, db 
\nonumber \\
&= \textstyle \int_0^{\infty} \big[ \int p(y \mid b,\sigma^2) \, 
p(b \mid t) \, db \big] \times p(t) \, dt \nonumber \\
&= \textstyle \int_0^{\infty} p(y \mid \sigma^2,t) \, p(t) \, dt, 
\label{eqn:nm-mixing-density} 
\end{align}
where
\begin{align}
p(y \mid \sigma^2,t) &= 
\textstyle \int p(y \mid b,\sigma^2) \, p(b \mid t) \, db \nonumber \\
&= \frac{1}{2t}
\left[\Phi\left(\frac{t-y}{\sigma}\right) + 
\Phi\left(\frac{t+y}{\sigma}\right) - 1 \right],
\label{eqn:marg-prob-y}
\end{align}
and where $\Phi(x)$ denotes the normal cumulative distribution
function. Note that, from \eqref{eqn:nmf}, $p(y \mid b, \sigma^2) =
N(y; b, \sigma^2)$. Since $\Phi(t + y) + \Phi(t - y)$ is
non-increasing in $y \in (0,\infty)$ for any $t \geq 0$, $p(y \mid
\sigma^2, t)$ is also non-increasing in $y \in (0, \infty)$ for any $t
\geq 0$. This implies that $p(y \mid \sigma^2, t)$ is unimodal with a
mode at zero, and therefore $p(y \mid f, \sigma^2)$ must also be
unimodal with a mode at zero since it is a mixture of unimodal
distributions that all have modes at zero. 

From \eqref{eqn:tweedie}, which was obtained by applying Tweedie's
formula, we have that
\begin{equation*}
S_{f, \sigma}(y) = y + \sigma^2 \ell_{\mathrm{NM}}'(y; f, \sigma) \leq y,
\end{equation*} 
in which the inequality is obtained by noting that
$\ell_{\mathrm{NM}}(y; f, \sigma)$ is non-increasing in $y \in (0,
\infty)$.  And since $\ell_{\mathrm{NM}}(y; f, \sigma^2)$ is
symmetric about zero, the shrinkage operator must be an odd
function; {\em i.e.}, $S_{f, \sigma}(y) = -S_{f, \sigma}(-y)$.

It remains to show that $S_{f, \sigma}(y)$ is non-decreasing on $y \in
\mathbb{R}_{+}$. Since $p(y \mid b, \sigma^2) = N(y; b, \sigma^2)$ and
$p(b \mid t)$ is the uniform distribution on interval $[-t, t]$, the
posterior density is truncated normal:
\begin{align*}
p(b \mid y, \sigma^2, t) &\propto 
p(y \mid b, \sigma^2) \, p(b \mid t) \nonumber \\
&= N_{[-t,t]}(b; y, s^2)
\end{align*}
where $N_{[-t,t]}(x; \mu, s^2) = N(x; \mu, s^2) \, \mathbb{I}\{|x| <
t\}$ denotes the probability density of the normal distribution with
mean $\mu$ and variance $s^2$ truncated to the interval $x \in [-t,
  t]$. The expected value of the truncated normal,
\begin{equation*}
\mathbb{E}[X] = \mu + s \times 
\frac{N(-t; \mu, s^2) - N(t; \mu, s^2)}
     {\Phi((t - \mu)/s) - \Phi(-(t + \mu)/s)},
\end{equation*}
is non-decreasing with respect to $\mu$, for all $\mu \in \mathbb{R}$,
$t > 0$. To show this, the derivative of the
expected value with respect to $\mu$ is always positive:
$\frac{\partial}{\partial \mu} \mathbb{E}[X] = \mathrm{Var}(X)/s^2 >
0$. Therefore, $S_{f, \sigma}$ is a mixture of non-decreasing
functions on $\mathbb{R}_{+}$.
\end{proof}

\def\section{\@startsiction{section}{1}{\z@}{-0.24in}{0.10in}
             {\large\bf\raggedright}}

\bibliography{mr_ash}

\end{document}